\documentclass[11pt,final]{article}
\usepackage{mathrsfs}
\usepackage{mathtools}%
\mathtoolsset{showonlyrefs}

\usepackage{amsthm}
\usepackage{amsxtra}%
\usepackage{amsfonts}%
\usepackage{amssymb}%
\usepackage[a4paper, margin=2.5cm]{geometry}
\usepackage{orcidlink}

\usepackage{mleftright}
\mleftright

\usepackage{doi}

\providecommand{\etalchar}[1]{$^{#1}$}

\usepackage{color}
\usepackage{graphicx}
\usepackage{subfig}
\usepackage{array}
\usepackage{booktabs}
\usepackage{bbm}
\usepackage{enumitem}
\usepackage{dsfont}
\usepackage{authblk}
\usepackage{todonotes}
\usepackage{svg}
\usepackage[margin=1cm]{caption}

\usepackage{tikz}
\usetikzlibrary {graphs,graphs.standard}

\usepackage{zref-clever}
\zcsetup{nameinlink, cap, noabbrev}

\usepackage{hyperref}
\hypersetup{colorlinks,breaklinks,
	linkcolor=blue,urlcolor=blue,
	anchorcolor=blue,citecolor=blue}

\newtheorem{theorem}{Theorem}
\newtheorem{lemma}[theorem]{Lemma}
\AddToHook{env/lemma/begin}{%
	\zcsetup{countertype={theorem=lemma}}}
\zcRefTypeSetup{lemma}{Name-sg=Lemma}

\newtheorem{corollary}[theorem]{Corollary}

\newtheorem{proposition}[theorem]{Proposition}
\AddToHook{env/proposition/begin}{%
	\zcsetup{countertype={theorem=proposition}}}
\zcRefTypeSetup{proposition}{Name-sg=Proposition}

\zcRefTypeSetup{item}{Name-sg=Item,Name-pl=Items}

\newtheorem{fact}[theorem]{Fact}
\AddToHook{env/fact/begin}{%
	\zcsetup{countertype={theorem=fact}}}
\zcRefTypeSetup{fact}{Name-sg=Fact}

\theoremstyle{definition}
\newtheorem{remark}[theorem]{Remark}
\AddToHook{env/remark/begin}{%
	\zcsetup{countertype={theorem=remark}}}
\zcRefTypeSetup{remark}{Name-sg=Remark}
\newtheorem{definition}[theorem]{Definition}
\AddToHook{env/definition/begin}{%
	\zcsetup{countertype={theorem=definition}}}
\zcRefTypeSetup{definition}{Name-sg=Definition}

\newtheorem{assumption}[theorem]{Assumption}
\AddToHook{env/assumption/begin}{%
	\zcsetup{countertype={theorem=assumption}}}
\zcRefTypeSetup{assumption}{Name-sg=Assumption,Name-pl=Assumptions}

\numberwithin{equation}{section}
\numberwithin{theorem}{section}

\DeclareMathOperator{\supp}{supp}
\let\deg\relax
\DeclareMathOperator{\deg}{deg}

\DeclareMathOperator{\diam}{diam}
\DeclareMathOperator{\dist}{dist}

\DeclareMathOperator{\per}{per}
\DeclareMathOperator{\Id}{Id}

\DeclareMathOperator{\irr}{irr}

\DeclareMathOperator{\Bin}{Bin}

\newcommand{\sprod}[2]{\langle #1\rangle_{#2}}

\newcommand{\set}[1]{\left\{#1\right\}}
\newcommand{\norm}[1]{\left\|#1\right\|}
\newcommand{\abs}[1]{\left\lvert #1 \right\rvert}
\newcommand{\expa}[1]{\exp\left( #1 \right)}

\newcommand\restr[2]{{% we make the whole thing an ordinary symbol
		\left.\kern-\nulldelimiterspace % automatically resize the bar with \right
		#1 % the function
		\vphantom{\big|} % pretend it's a little taller at normal size
		\right|_{#2} % this is the delimiter
}}

\renewcommand{\bar}[1]{{\overline{#1}}}
\newcommand{\eps}{\varepsilon}
\renewcommand{\d}{\,\mathrm{d}}
\newcommand{\dx}{\d x}

\newcommand{\grad}{\nabla}
\newcommand{\gradper}{\grad^{\per}}

\newcommand{\G}{\mathcal{G}}
\newcommand{\J}{\mathcal{J}}
\renewcommand{\P}{{\mathbbm P}}
\newcommand{\E}{{\mathbbm E}}
\newcommand{\V}{{\mathcal V}}

\newcommand{\R}{\mathbbm{R}}
\newcommand{\N}{\mathbbm{N}}
\newcommand{\Q}{\mathbbm{Q}}
\newcommand{\Z}{\mathbbm{Z}}
\newcommand{\I}{\mathbbm{1}}
\renewcommand{\S}{\mathbbm{S}}

\renewcommand{\tilde}{\widetilde}
\newcommand{\phntm}{\phantom{{}={}}}

\let\emptyset\varnothing
\title{Asymptotic long-range order for the $XY$-model on random geometric  graphs}

\newcommand*\samethanks[1][\value{footnote}]{\footnotemark[#1]}
\author{
	Margherita Disertori\,\orcidlink{0009-0006-9854-3127}\thanks{Hausdorff Center for Mathematics and Institute for Applied Mathematics, University of Bonn.
	\\\mbox{}\hspace{0.59cm}\texttt{mdiserto@uni-bonn.de}, \texttt{mihailescu@iam.uni-bonn.de}} 
	\ and 
	Max Mihailescu\,\orcidlink{0009-0002-7382-2390}\samethanks[1]}

\renewcommand{\V}{{\mathcal{V}}}
\renewcommand{\E}{{\mathcal{E}}}

\begin{document}
	\maketitle
	
\begin{abstract}
	We study the classical $XY$-model on \emph{random geometric graphs} $\G_{n, \eps}$, which are obtained by sampling $n \in \N$ independent points in a finite domain $\Omega \subset \R^d$, $d \geq 2$, and connecting two points by and edge if their distance is of order $\eps > 0$. We refer to $\G_{n, \eps}$ as the \emph{random environment}. Letting $\eps \to 0$ as $n \to \infty$ at a sufficiently slow rate, these graphs capture the geometry of $\Omega$. 
	
	Denoting the inverse temperature  by $\beta$, we show that in the limit $\beta \to \infty$ at a rate depending on $n$ and $\eps$, the $XY$-model on $\G_{n, \eps}$ exhibits long range order in the sense that we prove a lower bound away from zero on the two-point function. Our result is \emph{quenched} in the random environment: long-range order holds with large probability, converging to one as $n \to \infty$. To prove the statement, we show that with high probability the environment is sufficiently regular to apply a convexity argument and the Brascamp--Lieb inequality. 
	
\end{abstract}

%%%%%%%%%%%%%%%%%%%%%%%%%%%%%%%%%%%%%%%%%%%%%%%%%%%%%%
\section{Introduction}
%%%%%%%%%%%%%%%%%%%%%%%%%%%%%%%%%%%%%%%%%%%%%%%%%%%%%

The $O(N)$-models are a family of classical spin model for ferromagnetic systems with \emph{continuous rotational symmetry}, where the spins take values in the $N{-}1$-dimensional unit sphere. We focus on the case $N=2$, also known as the \emph{$XY$}- or \emph{plane rotator model}.
For a configuration $S \colon \Z^d \to \S^{1}$, its formal Hamiltonian is given by
\begin{equation}
	\mathscr H(S) 
	:= - \beta \sum_{\set{i,j} \in \E} S_i \cdot S_j
	= - \beta \sum_{\set{i,j} \in \E} \cos(\theta_i - \theta_j),
\end{equation}
where we parametrize $S_i = (\cos\theta_i, \sin \theta_i)$ by its angle $\theta_i \in [-\pi, \pi)$, and where $\E$ are the edges of the underlying graph.
In their seminal work, Fröhlich--Simon--Spencer \cite{Frohlich:InfraredBoundsPhase1976} proved that on the $d$-dimensional lattice these models exhibit \emph{long-range order} at sufficiently low temperatures in three and more dimension, whereas the Mermin--Wagner theorem \cite{Mermin:AbsenceFerromagnetismAntiferromagnetism1966} (cf. also \cite{Pfister:SymmetryGibbsStates1981}, \cite[Chapter 9]{Friedli:StatisticalMechanicsLattice2017}) precludes continuous symmetry breaking in two-dimensional lattices. However, in two dimensions the $XY$-model undergoes the BKT phase transition, which was proved in the celebrated work of Fröhlich--Spencer \cite{Frohlich:KosterlitzThoulessTransitionTwodimensional1981}.
For a review on the $O(N)$-models we refer the reader to notes by Peled--Spinka \cite{Peled:LecturesSpinLoop2019}. 

In this article we consider the $XY$-model on a certain family of random graphs, so called \emph{random geometric graphs}, which are formally introduced in \zcref{sec:main_results}. Heuristically, one randomly samples $n$ independent points on a bounded domain $\Omega \subset \R^d$. These random points form the vertices of the graph, while the edges consist of all pairs of points which are at distance at most $\eps > 0$. As $n \to \infty$, one lets $\eps \to 0$ (at some minimal rate to ensure that the graph is connected), in order to better resolve the geometry of $\Omega$. We show that, annealed in the environment (the random graph), the $XY$-model on these graphs exhibits long-range order in all dimensions $d \geq 2$, provided one considers the zero temperature limit $\beta \to \infty$ at a sufficiently fast rate.

These random geometric graphs arise for example in semi-supervised learning approaches to machine learning, for a review see Calder--Drenska \cite{Calder:ConsistencySemisupervisedLearning2024}.
Many results on the limit $n \to \infty$ have been obtained in the literature. For example, Hein--Audibert--von Luxburg \cite{JMLR:v8:hein07a} obtained pointwise limits of the graph Laplacian to the weighted Laplace--Beltrami operator of $\Omega$, while Calder--García Trillos \cite{Calder:ImprovedSpectralConvergence2022} proved convergence of eigenvalues and eigenvectors.
Further, Bungert et al. \cite{Bungert:ConvergenceRatesPoisson2024} show the convergence of solutions to the graph Poisson equation to solutions to the (weighted) Poisson equation in $\Omega$ with Neumann boundary conditions.

In the recent years there has been increasing interest in the $XY$-model on random graphs and we give a short overview here. Dario--Garban \cite{Dario:PhaseTransitionsXY2025} consider the $XY$-model on the infinite cluster of Bernoulli percolation for $p>p_c$ and on Poisson--Voronoi graphs. They show that the disorder does not affect the phase transitions too much, in the sense that the BKT transition persist in two dimensions, whereas continuous symmetry breaking still occurs in three and more dimensions. Dario--van Engelenburg--Garban \cite{Dario:ImpactDisorderNonconvex2026} extend these results also to the closely related Villain model. 
A further example stems from the classical Heisenberg model (another name for the $O(3)$-model), which, conditioned on its third spin coordinate, yields an $XY$-model in random conductances (i.e. on a random graph). Patrascioiu--Seiler \cite{Patrascioiu:PercolationExistenceSoft2002} argue for the existence of a massless phase in the two-dimensional $O(3)$-model by examining the percolation properties of the resulting random $XY$-system, thereby arguing against a widely believed conjecture, which states that the tow-dimensional $O(3)$-model should be massive at all temperatures.
Aru--Garban--Sepúlveda \cite{Aru:Percolation2DClassical2025} make part of their argument rigorous, but also provide a counterexample to a different part of their argument, thereby shedding new light on the conjecture.

Our results concern graphs with a spacing $\eps \to 0$, so let us briefly remark on some results on scaling limits for the $XY$-model. On the lattice $\eps\Z^2$, Alicandro--Cicalese \cite{Alicandro:VariationalAnalysisAsymptotics2009a} show with variational methods how the appearance of vortex-like singularities can be described by a proper scaling of the energy when $\eps \to 0$. Effectively, this is a zero temperature description of configurations which have an energy cost comparable to the ground state. Moreover, Newman--Wu \cite{Newman:GaussianFluctuationsClassical2018} consider a zero temperature limit $\beta \to \infty$ as $\eps \to 0$ at a sufficiently fast rate, and prove that the fluctuation field of the lattice $XY$-model converges to standard Gaussian white noise. Our results are in a similar spirit, as we show that if $\beta \to \infty$ at a sufficiently fast rate, the $XY$-model on random geometric graphs is ordered as the system size $n \to \infty$. It is an open question weather the random graph $XY$-measures converge to some limiting measure on $\Omega$.

An additional motivation of our work is to further the understanding of the transition from disorder to long-range order. The original proof of of symmetry breaking by Fröhlich--Simon--Spencer is based on reflection positivity and the infrared bound, which are very powerful tools but have the shortcoming of requiring a high degree of symmetry of the underlying graph. 
Subsequently, there has been considerable effort in providing alternatives.
Fröhlich--Spencer \cite{Frohlich:MasslessPhasesSymmetry1982a} provide a proof of symmetry breaking of abelian spin systems in three and more dimensional lattices, which includes the $XY$-model. More recently Garban--Spencer \cite{Garban:ContinuousSymmetryBreaking2022a} gave yet a different proof based on a comparison with a certain $XY$-model with (Nishimori) disorder. Further, Balaban \cite{Balaban:LowTemperatureExpansion1995,Balaban:LowTemperatureExpansion1996} developed an approach to the $O(N)$-models based on the renormalization group. 
The tools we develop in the present article are part of these efforts to provide alternatives to reflection positivity. 

\subsection{Setting and main results}
\label{sec:main_results}

Consider a triple $G = (V, E, J)$, where $V$ is a finite set, $E \subset \set{\set{i,j} \subset V}$ and $J \colon V \times V \to [0, \infty)$ is a symmetric function. We will always assume that
\begin{equation}\label{eq:assumption_edges_weights}
		J(i,j) > 0 \text{ if and only if } \set{i,j} \in E.
\end{equation}
	We call the triple $G =(V, E, J)$ a \emph{(weighted) graph}. If a statement does not depend on the edge weights $J$ we will simply write $G=(V,E)$. We assume throughout this paper that all graphs are finite and simple, i.e., no vertex is connected to itself and there is at most a single edge between vertices.
	
	We define the \emph{$XY$-model} on $G$ by the following probability measure:
	\begin{equation}\label{eq:def_$XY$_on_G}
		\d \mu_{G;\beta}(\theta) := Z_{G;\beta}^{-1} \expa{\beta \sum_{\set{i,j} \in E} J(i,j) \cos(\theta_i - \theta_j)} \d\theta_{V},
	\end{equation}
	where $\d \theta_V$ is the Lebesgue measure on $[-\pi, \pi)^{V}$ and $Z_{G;\beta}$ is a normalization constant called the \emph{partition function}. In this context, $J(i,j)$ are the \emph{coupling constants}.
	We write the expectation with respect to $\mu_{G;\beta}$ as $\sprod{\cdot}{G; \beta}$.
	
	We introduce a second source of randomness which we refer to as the \emph{random environment}. Let $\Omega \subset \R^d$, $d\geq 2$, open and bounded and $(\Xi, \mathcal F, \P)$ a probability space. Let $X_1, \ldots, X_n \colon \Xi \to \Omega$ be independent random variables with probability density $\rho$ on $\Omega$. We set $\V_n := \set{X_1, \ldots, X_n}$. For a function $\eta \colon [0, \infty) \to \R$ define $\eta_\eps := \eps^{-d}\eta\left(\frac{\cdot}{\eps}\right)$. 
	
	In this work we focus on random coupling constants $\J_{n,\eps} \colon \V_n \times \V_n \to [0, \infty)$ given by
	\begin{equation}\label{eq:random_coupling_constants_normalized}
		\J_{n,\eps}(x, y) := 
			\frac{\eta_\eps(\abs{x - y})}{\sqrt{\deg_{n,\eps}(x) \deg_{n,\eps}(y)}}, \quad x,y \in \V_n,
	\end{equation}
	where we define the \emph{degree} $\deg_{n,\eps} \colon \Omega \to [0, \infty)$ by
	\begin{equation}\label{eq:def_degree}
		\deg_{n,\eps}(x) := \sum_{y \in \V_n} \eta_\eps(\abs{x - y}), \quad x \in \Omega.
	\end{equation}
	Note that in this definition we do not require $x \in \V_n$ and for any $x \in \V_n$ we have $\deg_{n,\eps}(x) \geq \eta(0) > 0$. Moreover, $\J_{n,\eps}$ is invariant under transformation of $\eta \mapsto c \eta$.

	Further, we define
	\begin{equation}
		\E_{n,\eps} := \set{\set{x,y} \subset \V_n \mid \J_{n,\eps}(x,y) > 0}.
	\end{equation}
	Then $\G_{n,\eps} := (\V_n, \E_{n, \eps}, \J_{n,\eps})$ is a random graph, called \emph{random geometric graph}. Since $\V_n$ is a random object, the weights $\mathcal J_{n,\eps}$ depend on the realization of the $X_i$. When we want to make clear that we consider a fixed realization $\omega \in \Xi$ of the random environment, we will write $\G_{n,\eps}(\omega) \equiv (\V_n(\omega), \E_{n,\eps}(\omega), \J_{n,\eps}^\omega)$. Similarly, we will sometimes write $\deg_{n,\eps}^\omega$ to make clear that the degree function is considered for a fixed $\omega \in \Xi$.
	
	\begin{remark}
		Let us briefly motivate our choice of coupling constants \eqref{eq:random_coupling_constants_normalized}. If one chooses
		\begin{equation}
			\tilde\J_{n,\eps}(x, y) := 
			\frac{\eta_\eps(\abs{x - y})}{\sigma_\eta n \varepsilon^2}, \quad x,y \in \V_n,
		\end{equation}
		then the corresponding graph Laplacian $-\Delta_{\tilde\G_{n, \eps}}$ (see \eqref{eq:def_graph_laplacian} for a precise definition) is known to converge to (a weighted version of) the Laplace--Beltrami operator on $\Omega$ (cf. \cite{Calder:ImprovedSpectralConvergence2022}), so that it is a natural choice for the coupling constants. In this article we replace the normalization $\frac1n$ by $\frac1{\sqrt{\deg_{n,\eps}(x)}\sqrt{\deg_{n,\eps}(y)}}$ to emphasize the local structure of the graph. However, since asymptotically $\deg_{n,\eps} \sim \frac1n$ (cf. \zcref{thm:degrees_bounded}), we are considering the same scaling regime.
	\end{remark}
	
	Throughout the paper we assume $d\geq 2$ and that $\Omega\subset \R^{d}$ is open and bounded. We will need in addition the following assumptions on $\eta$, $\Omega$ and $\rho$.
	
	\begin{assumption}\label{ass:eta}
		The kernel function $\eta \colon [0,\infty) \to [0, \infty)$ is continuous at the origin, non-increasing and satisfies $\eta(0) > 0$. Moreover, we assume that
		\begin{equation}
			\label{eq:def_sigma_eta}
			\sigma_{\eta}:= \int_{\R^d} \abs{z \cdot e_1}^2 \eta(\abs{z}) \d z=
			{w_d}\int_{0}^\infty t^{d+1} \eta(t) \d t< \infty,
		\end{equation}
		where $w_d$ is the volume of the $d$-dimensional unit ball.
	\end{assumption}
	
	\begin{assumption}\label{ass:Omega}
		 $\Omega$  has Lipschitz boundary.
	\end{assumption}
	
	\begin{assumption}\label{ass:rho}
		The density $\rho \in L^\infty(\Omega)$ is Lipschitz continuous and there exist constants such that $0 < \rho_{\min} \leq \rho \leq \rho_{\max}$.
	\end{assumption}
	
	Our main result is a quenched lower bound on the two-point function of the $XY$-model on $\G_{n, \eps}$. With high probability in the random environment, tending to $1$ as $n\eps^d \to \infty$, we obtain that the model on these graphs exhibits long range order, if $\beta \to \infty$ sufficiently fast. 
	
	\begin{theorem}\label{thm:main}
		Assume that $\Omega$, $\eta$ and $\rho$ satisfy \zcref{ass:Omega,ass:eta,ass:rho}.
		There exist positive constants $c_1$, $c_2$, $c_3$, $c_4$ and $c_5$ depending only on $\Omega$, $\rho$ and $\eta$
		such that the event
		\begin{equation}\label{eq:assertion_main_thm}
			\text{for all } x,y \in \V_n  \colon 
			\sprod{\cos{\theta_{x} - \theta_{y}}}{\G_{n,\eps};\beta} 
			\geq 1 - \frac{c_3}{\beta \eps^2} - c_4 n^2 e^{-\frac{c_5\beta}{n\eps^d}}
		\end{equation}
		has probability at least $1 - c_1 n \expa{-c_2 n \eps^d}$.
	\end{theorem}

	This directly implies the following annealed version of our result.
	
	\begin{corollary}
		Assume that $\Omega$, $\eta$ and $\rho$ satisfy \zcref{ass:Omega,ass:eta,ass:rho}.
		There exist positive constants $c_1$, $c_2$, $c_3$, $c_4$ and $c_5$ depending only on $\Omega$, $\rho$ and $\eta$
		such for all $i,j \in \set{1, \dots, n}$
		\begin{equation}
			\mathbbm E\left[
			\sprod{\cos{\theta_{X_i} - \theta_{X_j}}}{\G_{n,\eps};\beta} \right]
			\geq 1 - \frac{c_3}{\beta \eps^2} - c_4 n^2 e^{- \frac{c_5\beta}{n\eps^d}} - c_1 n e^{-c_2 n \eps^d},
		\end{equation}
		where $\mathbbm E$ denotes the expectation with respect to $\P$.
	\end{corollary}
			
	\subsection{Proof strategy and outline of the paper}
	
	In \zcref{sec:preliminaries} we recall some known results on the $XY$-model and on random geometric graphs.
	We also  introduce some algebraic notions on graphs, in particular the concept of  \emph{algebraic irreducibility} which measures  the size of its ``largest hole''. 
		\zcref{sec:finite_graphs} is devoted to proving the following lower bound of the two-point function of the $XY$-model on a deterministic graph $G= (V,E,J)$ in terms of the effective resistance $\mathbf R_G$ (see \eqref{eq:def_effective_resistance}) 
	\begin{equation}\label{eq:bound_intro_deterministic}
		\sprod{\cos{\theta_{x} - \theta_{y}}}{G}
		\geq 1
		- \tfrac{c_0 \mathbf R_G}{\beta}
		- c_1 \abs{E} \, e^{-c_2 \beta},
	\end{equation}
	where  $c_{1}$ is a positive constant,  $c_0=c_{0} (\irr G)$ and  $c_2=c_{2} (\irr G,J_{0})$, with  $J_0 = \min_{\set{i,j} \in E} J(i,j)$.
	This bound is extended to the random setting  in \zcref{sec:order_rand_geometric_graphs}.
	The main difficulty is to show  that the following event has a large probability:
	\begin{equation}
		\set{\irr \G_{n, \eps} \leq K \text{ and } \mathbf R_{\G_{n, \eps}} \leq R \text{ and } J_0(\G_{n,\eps}) \geq \bar{J}},
	\end{equation}
	with constants $K$, $R$ and $\bar{J}$ depending on $n$ and $\eps$ but not on $\beta$. For all realizations of the random
	environment within this event, $c_{0} (\irr \G_{n, \eps})\leq \bar{c}_{0}$ and  $c_{2} (\irr G,J_{0})\leq \bar{c}_{2}$, hence the bound \eqref{eq:bound_intro_deterministic} implies
	\eqref{eq:assertion_main_thm}. Controlling the irreducibility is the content of \zcref{sec:_irr_random_graphs}, while
	the effective resistance and the smallest coupling constant are dealt with in \zcref{sec:random_graphs}.
	To conclude this article, we prove a similar result in \zcref{sec:interior}, which holds for all $x,y \in \V_n$ that are sufficiently far away from
	$\partial \Omega$. We chose to give this argument, because it has a simple proof based on a  comparison between the $XY$-model on random graphs and
	on the lattice $\Z^{d},$ $d\geq 3$. For the latter we use a finite volume result from  \cite{Garban:ContinuousSymmetryBreaking2022a}.

	\section{Preliminaries}
	\label{sec:preliminaries}
%%%%%%%%%%%%%%%%%%%%%%%%%%%%%%%%%%%%%%%%%%%%%%%%%%%%%%%%%%%%%%%%%%%%%%%%%%%%%%

	We introduce a few concepts and results which we will need through this article. While they are standard in the literature, we choose to give short explanations and proofs to keep the article a bit more self-contained.

%%%%%%%%%%%%%%%%%%%%%%%%%%%%%%%%%%%%%%%%%%%%%%%%%%%%%%%%%%%%%%%%%%%%%%%
	\subsection{Facts about the \texorpdfstring{$XY$}{XY}-model}
	
	We begin with two standard facts about the classical $XY$-model.
	The following is a classical application of Ginibre's inequality and allows to relate two models with different coupling constants.
	\begin{fact}[Monotonicity in the $XY$-weights]\label{thm:decrease_coupling}
		For any finite set $\Lambda$ and $x,y \in \Lambda$ and $\beta > 0$, consider the function $f \colon [0,\infty)^{V \times V} \to \R$ given by
		\begin{equation}
			[0, \infty)^{\Lambda \times \Lambda} \ni J \mapsto \sprod{\cos \theta_x - \theta_y}{(\Lambda, E_J, J); \beta},
		\end{equation}
		where
		\begin{equation}
			E_J := \set{\set{i,j} \subset \Lambda \mid J(i,j) > 0}.
		\end{equation}
		Then for any $i,j \in \Lambda$
		\begin{equation}
			\frac{\partial}{\partial J_{uv}} f(J) \geq 0.
		\end{equation}
		In particular, if $J, J' \in [0,\infty)^{\Lambda \times \Lambda}$ such that $J \leq J'$ coordinate-wise, then
		\begin{equation}
			\sprod{\cos \theta_x - \theta_y}{(\Lambda, E_{J'}, J'); \beta} 
			\geq \sprod{\cos \theta_x - \theta_y}{(\Lambda, E_J, J); \beta}. 
		\end{equation}
	\end{fact}
	\begin{proof}
		We have
		\begin{equation}
			\begin{aligned}
				&\phntm \frac{1}{\beta}\frac{\partial}{\partial J_{uv}} f(J) \\
				&= \left(\sprod{\cos\left(\theta_x - \theta_y\right) \cos\left(\theta_i - \theta_j\right)}{(\Lambda, E_J, J); \beta} - \sprod{\cos \theta_x - \theta_y}{(\Lambda, E_J, J); \beta} \sprod{\cos \theta_i - \theta_j}{(\Lambda, E_J, J); \beta}\right)
				\geq 0
			\end{aligned}
		\end{equation}
		by Ginibre. Thus $f$ is increasing in each coordinate and the claim follows.
	\end{proof}
		
	The second result states that, in a finite graph, nearest neighbors are aligned with very large probability.
	\begin{fact}[Nearest neighbor alignment, {\cite[Theorem 2]{Bricmont:CorrelationInequalitiesContour1981}}]\label{th:NNalign}
			Let $0 < \kappa < \pi$. Let $G=(V, E, J)$ be a finite and connected graph and let
			\begin{equation}\label{def:J0}
				J_0 := \min\set{J(i,j) \mid \set{i,j} \in E} > 0.
			\end{equation}
			Then there exists a positive constant $c_0$ independent of $G$ and $\kappa$ such that for all $\beta > 0$
			\begin{equation}
				\mu_{G;\beta}\Big( \exists \{u,v\}\in E \colon  \cos (\theta_{u}-\theta_{v})\leq \cos \kappa   \Big) 
				\leq \abs{E}c_0\, e^{-\frac{\beta J_0 \kappa^2}{\pi}}.
			\end{equation}
	\end{fact}

%%%%%%%%%%%%%%%%%%%%%%%%%%%%%%%%%%%%%%%%%%%%%%%%%%%%%%%%%%%%%%%%%%%%%%%
	\subsection{Graph-based calculus}
	
	We further need a few definitions from graph-based calculus. 
	Let $G=(V, E, J)$ be a finite, undirected graph. For functions $u,v \colon V \to \R$ consider the euclidean scalar product
	\begin{equation}
		\left(u, v\right) := \sum_{x \in V} {u(x)}v(x)
	\end{equation}
	which induces a norm $\norm{u}_2^2 := \left(u, v\right)$. Define moreover
	\begin{equation}
		\ell^2(G) := \set{u \colon V \to \R \mid \norm{u}_{2} < \infty},
	\end{equation}
	the space of real valued functions on $V$ with the euclidean scalar product. 
	The (weighted) graph Laplacian $-\Delta_G\! \colon\! \ell^2(G) \!\to\! \ell^2(G)$  is defined by
	\begin{equation}\label{eq:def_graph_laplacian}
		\left(-\Delta_G u\right)(x) := \sum_{\set{x,y} \in E} J(x,y)\,  (u(x) - u(y)) \quad \text{for } x \in V.
	\end{equation}
By discrete integration by parts  we have
	\begin{equation}\label{eq:def_graph_laplacian1}
		 \left(v, -\Delta_{G} v \right)
		 := \sum_{\set{i,j} \in E} J(i,j)\, \left({v(i) - v(j)}\right)^2,
	\end{equation}
	which implies in particular
	\begin{equation}\label{eq:def_graph_laplacian2}
		\ker\left(-\Delta_G\right) = \set{u \in \ell^2(G) \mid u \equiv \text{const}}.
	\end{equation} 

%%%%%%%%%%%%%%%%%%%%%%%%%%%%%%%%%%%%%%%%%%%%%%%%%%%%%%%%%%%%%%%%%%%%%%%
	\subsection{Algebraic representation of loops in graphs}
	\label{sec:cycle_space}
	
	Our proof of \zcref{thm:main} relies on the fact that with high probability the graph $\G_{n,\eps}$ does not have large ``holes''. This means that any loop can be decomposed into a sum over small loops of a maximal size. 
	To formalize this, it is useful to introduce the concept of a \emph{cycle basis} of a graph. In this setting, the size of the largest hole is
	described by the concept of \textit{algebraic irreducibility} (see \eqref{eq:defirr}). 
	We adopt definitions commonly used in graph theory (cf. \cite{Diestel:GraphTheory2025,Kavitha:CycleBasesGraphs2009}). 
%%%%%%%%%%%%%%%%%%%%%%%%%%
	\begin{definition}[Paths and cycles]
		Let $G=(V, E)$ be an undirected graph. A \emph{path} $P=x_0 \cdots x_n$ is a sequence of vertices of $V$, such that $\set{x_{k-1}, x_k} \in E$. We allow a  path to intersect itself, however we do not allow the repetition of edges. We denote the length of $P$ by $\abs{P} := n$, which corresponds to the number of edges.
		We say that $P$ is a \emph{cycle} if $x_0 = x_n$.
		Note that every path induces a subgraph of $G$ by the vertices and edges it traverses.
	\end{definition}

%%%%%%%%%%%%%%%%%%%%%%%%%%%%%%%%%%%%%%%%%%%%%%%%%%%%%%%%%%%%%%%%%%%%%%%%%%%%%%%%%%%
	\begin{definition}[Edge and cycle space]
		Let $G=(V, E)$ be an undirected graph and let $\mathbb F_2 \equiv \Z /2 \Z$  be the field over $\set{0,1}$.

The \emph{edge space} of $G$, $\mathbf E \equiv \mathbf E(G)$, is the vector space 
		 of functions $f \colon E \to \mathbb F_2$. while the \emph{cycle space} $\mathbf C \equiv \mathbf C(G)$ is the subspace of $\mathbf E$ of all functions $f \colon E \to \mathbb F_2$ satisfying the constraint
		\begin{equation}\label{eq:cyclecond}
			\sum_{j\colon\set{i,j} \in E} f(\set{i,j}) = 0 \quad \text{for all } i \in V.
		\end{equation}
		For $f \in \mathbf E(G)$ we set 
		\begin{equation}
			\abs{f} := \sum_{e \in E} f(e).
		\end{equation}
	\end{definition}
%%%%%%%%%%%%%%%%%%%%%%%%%%%%%%%%%%%%%%%%%%%%%%%%%%%%%%%%%%%%%%%%%%%%%%%%%%%%%%%%%%%%%%%%	
	\begin{remark}\label{rem:cyclespace}
A function $f$ in $ \mathbf E(G)$ identifies a unique subset of edges $A_{f}\in E$ of size $|A_{f}|=|f|$.
Conversely, each subset  $A\subset E$ is represented by  
its indicator function $\I_A\in  \mathbf E(G)$. In particular, given $A,B \subset E$, $\I_{A} + \I_B \in \mathbf E(G)$ represents the element $A \triangle B \subset E$. For $A \subset E$, we will often write $A \in \mathbf E(G)$ instead of $\I_A \in \mathbf E(G)$. 
		
 Note that \eqref{eq:cyclecond} states that each vertex has an even number of incident edges. This implies in particular that every element in the cycle space $\mathbf C(G)$ corresponds to a collection of edge disjoint cycles in the graph $G$. 
		
	\end{remark}
	
	\begin{definition}[Cycle basis and irreducibility]
		Let $G=(V, E)$ be an undirected graph. A basis of the vector space $\mathbf C(G)$ is called a \emph{cycle basis} of $G$. 
		By a slight abuse of notation we will call $\set{A_1, \ldots A_n}$, $A_i \subset E$ a cycle basis of $G$, if $\set{\I_{A_1}, \ldots \I_{A_n}}$ is a cycle basis of $G$. 
		
		We define the \emph{algebraic irreducibility} of $G$ as
		\begin{equation}\label{eq:defirr}
			\irr(G) := \min_{\substack{\set{A_1, \ldots, A_n}\\\text{ is a cycle basis of $G$}}} \max_{i=1}^n \abs{A_i}.
		\end{equation}
	\end{definition}
	%%%%%%%%%%%%%%%%%%%%%%%%%%%%%%%%%%%%%%%%%%%%%%%%%%%%%%%%%%%%%%%%%%%%%
	\begin{remark}\label{rem:irrG}
		
	The number $\irr(G)$ is the size of the largest 'hole' in the graph $G$. More formally, it is the size of the longest cycle which cannot be decomposed into
	a sum of shorter cycles. Since the length of a cycle is  at least 3, it holds $\irr(G)\geq 3$.
	Here are three examples. 
	
	\begin{enumerate}[label=(\roman*)]
				\item Consider the torus $G=\left(\Z / L\Z\right)^d$ with nearest neighbor edges. In this case, $\irr(G) = 4$, since the collection of unit squares forms a cycle basis and there are no cycles of length smaller that four in $G$.
				\item If $G$ is the honeycomb lattice in two dimensions, then $\irr(G) = 6$, since the collection of hexagons forms a cycle basis and there are no cycles of length smaller than six in $G$.
				\item The complete graph $K_m$ with $m \geq 3$ vertices $v_1, \dots, v_m$ has $\irr(K_m) = 3$ since the collection of cycles $v_1 v_i v_{i+1} v_1$ for $i \in \set{2, \dots, m-1}$ is a cycle basis.
	\end{enumerate}
	To the best of our knowledge the algebraic irreducibility of a graph has not been studied as a graph invariant in the literature. However, there has been interest in finding cycle bases which minimize $\max_{i=1}^n \abs{A_i}$. This problem is known as the \emph{Shortest Maximal Cycle Basis} problem (SMCB) and can be solved in polynomial time $\mathcal O(\abs{E}^3 \abs{V})$, see \cite{Chickering:FindingCycleBasis1995,Horton:PolynomialTimeAlgorithmFind1987}.
	\end{remark}
	%%%%%%%%%%%%%%%%%%%%%%%%%%%%%%%%%%%%%%%%%%%%%%%%%%%	

	\begin{definition}[Orientation]
	Let $G=(V, E)$ be an undirected graph. For all edges $e = \set{i,j} \in E$ choose an orientation $\vec e \in V \times V$, which is either $(i,j)$ or $(j,i)$. The directed graph $\vec G = \left(V, \vec E\right)$ with the directed edge set $\vec E = \set{\vec e \mid e \in E}$ is called \emph{orientation} of $G$.
	\end{definition}
	
	\begin{definition}[Directed cycle space]
		Let $\vec G=(V, \vec E)$ be a directed graph. 
		We call the vector space over $\Q$ of functions $\vec f \colon \vec{E} \to \Q$ satisfying the constraint
		\begin{equation}\label{eq:directed_cycle_contraint}
			\sum_{\vec e \in \delta^+(v)} \vec f(\vec e) = \sum_{\vec e \in \delta^-(v)} \vec f(\vec e) \quad \text{for all } i \in V.
		\end{equation}
		the \emph{directed cycle space} $\mathbf{C}(\vec G)$ of $\vec G$.
		Here, $\delta^+(v) := \set{(v, w) \in \vec E \mid w \in V}$ are the edges leaving $v$, and $\delta^-(v) := \set{(w, v) \in \vec E \mid w \in V}$ are the edges entering $v$.
	
	\end{definition}
	
	\begin{definition}
		Let $G = (V,E)$ be an undirected graph with orientation $\vec G$. 
		Let $\vec f \in \mathbf{C}(\vec G)$ such that $\vec f(\vec E) \subset \Z$. We define its \emph{projection} onto the cycle space of $G$ $\pi_{\vec f} \in \mathbf C(G)$ by
		\begin{equation}
			\pi_{\vec f}(e) := \vec f(\vec e) \mod 2, \quad e \in E.
		\end{equation}
		Let $g \in \mathbf C(G)$ and $\vec g \in \mathbf C(\vec G)$ such that $\vec g(\vec E) \subseteq \set{-1, 0, 1}$. We say that $\vec g$ is a \emph{lifting} of $g$ if $\pi_{\vec g} = g$.   Such a lifting always exists (cf. \zcref{thm:lifting_exists}).
	\end{definition}
	
	\begin{proposition}[cf. {\cite[Lemma 2.4]{Kavitha:CycleBasesGraphs2009}}]\label{thm:lifting_basis}
		Let $G = (V,E)$ be an undirected simple graph with orientation $\vec G$. Let $f_1, \ldots, f_n \in \mathbf C(G)$ be a basis of $\mathbf C(G)$ and let $\vec f_i \in \mathbf C(\vec G)$, $i \in \set{1, \dots, n}$ be a lifting of $f_i$. Then $\vec f_1, \ldots, \vec f_n$ is a basis of $\mathbf C(\vec G)$. 
	\end{proposition}
	For convenience of the reader, we include the proof of this result in  \zcref{appendix}.

	%%%%%%%%%%%%%%%%%%%%%%%%%%%%%%%%%%%%%%%%%%%%%%%%%%%%%%%%%%%%%%%%%%%%%%%
	\subsection{Regularity properties of random geometric graphs}
	
	It is known that random geometric graphs, in the limit $n \to \infty$, and $\eps \to 0$ at some minimal rate, approximate the continuum domain $\Omega$ well. In our application, we need the following regularity results, which hold with large probability: each point $x \in \Omega$ is at distance
$\mathcal O(\varepsilon )$ from the graph, the degrees $\deg_{n,\eps}$ scale proportional to $\frac1n$ (and in fact converge to the density $\rho$), and the graph satisfies a discrete Poincaré inequality, uniform in $n$ and $\eps$. We give the precise statements below.	
		
	Recall the definition of the random graph $\G_{n, \eps} = (\V_n, \E_{n,\eps}, \J_{n,\eps})$ from \zcref{sec:main_results}.
	For $h > 0$ we define the $h$-cube centered at $x \in h\Z^d$ to be
	\begin{equation}
		\square_h(x) := x + \left[-\frac{h}2, \frac{h}2\right)^d.
	\end{equation}
	Further, let 
	\begin{equation}\label{def:Lambda_h}
		\Lambda_h := h\Z^d \cap \Omega.
	\end{equation}
Later we will take $h$ proportional to $\varepsilon$.  The next lemma shows that each  $h$-cube contains at least $1\leq k\leq \mathcal( n h^d)$
points of $\V_n$ with large probability. Since the edges play no role in the statement, we only need to use  \zcref{ass:Omega,ass:rho}.
	\begin{lemma}\label{thm:each_square_has_k_points}
		Assume that $\Omega$ and $\rho$ satisfy \zcref{ass:Omega,ass:rho} respectively. Let $h > 0$ such that $n h^d \geq 1$. 
		Then there exists constants $c_1 \equiv c_1(\Omega)>0$, $c_2 \equiv c_2(\Omega, \rho_{\min})>0$ and $c_3 \equiv c_3(\Omega, \rho_{\min})>0$, such that for any $1\leq k \leq c_3 n h^d$,
		\begin{equation}
			\P\big(\text{for all } v \in \Lambda_h \colon \abs{\V_n \cap \square_h(v)} \geq k\big)
			\geq 1 - c_1 n \expa{-\frac{(c_2 n h^d - k+1)^2}{c_2 n h^d}}.
		\end{equation}
	\end{lemma}
	\begin{proof}
		We follow \cite[Lemma 5.20]{Calder:CalculusVariations2024} and assume without loss of generality that $\Omega \subset [0,\alpha_0]^d$, otherwise we translate.
		We use the following result: for $Z \sim \Bin(n, p)$ it holds
		\begin{equation}\label{eq:bin_bound}
			\P\left(Z \leq k\right) \leq \expa{-\frac{(np - k)^2}{2np}} \qquad \text{for all } 0\leq k \leq np.
		\end{equation} 
		For $v \in \Lambda_h$ let $Y_v := \abs{\V_n \cap \square_h(v)} = \sum_{i=1}^n \I_{X_i \in \square_h(v)}$ be the number of random points in $\square_h(v)$. Then
		\begin{equation}
			\mathcal A 
			:= \big\{\text{for all } v \in \Lambda_h \colon \abs{\V_n \cap \square_h(v)} \geq k\big\}
			= \bigcap_{v \in \Lambda_h} \set{Y_v \geq k  }.
		\end{equation}
		By independence of the trials we obtain that $Y_v \sim \Bin(n, p_v)$, where
		\begin{equation}
			p_v = \int_{\square_h(v) \cap \Omega} \rho \dx \geq c_0 \rho_{\min} h^d.
		\end{equation}
		The lower bound holds, for a constant $c_{0}=c_0(\Omega)\!>\!0$, because $\Omega$ is a Sobolev extension domain (cf. \cite{Hajlasz:SobolevEmbeddingsExtensions2008}).
			Applying \eqref{eq:bin_bound}, we can estimate for all $1\leq k \leq c_0 \rho_{\min} n h^d + 1$
		\begin{equation}
			\P(Y_v < k)=\P(Y_v \leq  k-1)
			\leq \expa{-\frac{(c_0 \rho_{\min} n h^d - k+1)^2}{c_0 \rho_{\min} n h^d}}.
		\end{equation}
		Hence,
		
\begin{equation}
	\P(\mathcal A) \geq 1 - \sum_{v \in V} \expa{-\frac{(c_0 \rho_{\min} n h^d - k+1)^2}{c_0 \rho_{\min} n h^d}} \geq 1 - \alpha_0^d h^{-d} \expa{-\frac{(c_0 \rho_{\min} n h^d - k+1)^2}{c_0 \rho_{\min} n h^d}} 
\end{equation}
		Using  $n h^d \geq 1$ we can further simplify the probability to
		\begin{equation}
			\P(\mathcal A) 
			\geq 1 - \alpha_0^d n \expa{-\frac{(c_0 \rho_{\min} n h^d - k+1)^2}{c_0 \rho_{\min} n h^d}} \\
		\end{equation}
		as claimed.
	\end{proof}

	We will need control on the random degrees $\deg_{n,\eps}$ defined in \eqref{eq:def_degree}. This result is a direct consequence of \cite[Proposition 3.26]{Bungert:ConvergenceRatesPoisson2024}, using that $\Omega$ is Lipschitz continuous.
	
	\begin{lemma}[]\label{thm:degrees_bounded}
		Assume that $\Omega$, $\eta$ and $\rho$ satisfy \zcref{ass:Omega,ass:eta, ass:rho} and let $\eps \in (0,1)$.
		There exist positive constants $c_1 \equiv c_1(\Omega, \eta, \rho_{\min})$, $c_2 \equiv c_2(\Omega, \eta, \rho_{\max})$ and $c_3 \equiv c_1(\eta(0), \rho_{\min}, \rho_{\max})$ such that 
		\begin{equation}\label{eq:degree_bound}
			\P\left(\text{for all } x \in \V_n \colon  n c_1 \leq \deg_{n,\eps}(x) \leq n c_2\right)
			\geq 1 - 2 n\expa{-c_3 n \eps^d}.
		\end{equation}
	\end{lemma}

		Lastly, we will use the following probabilistic discrete Poincaré inequality.
	\begin{proposition}[Discrete Poincare inequality, {\cite[Proposition 3.30]{Bungert:ConvergenceRatesPoisson2024}}]\label{thm:discrete_poincare}
		Assume that $\Omega$, $\eta$ and $\rho$ satisfy \zcref{ass:Omega,ass:eta, ass:rho}. Additionally, assume that $\supp \eta \subset [0,1]$ and that $\norm{\eta}_{L^1([0,1])} = 1$.
		
		There exist positive constants $c_0 \equiv c_0(\Omega, \eta, \rho)$, $c_1 \equiv c_1(\Omega)$, $c_2 \equiv c_2(\Omega, \eta(0), \rho_{\min}, \rho_{\max})$ and $\eps_0 \equiv \eps_0(\Omega, \rho, \eta)$, such that for any $n \in \N$, $n^{-\frac1d} < \varepsilon \leq \eps_0$ the event that
		\begin{equation}\label{eq:discrete_poincare}
			\norm{u - (u)_{\deg_{n,\eps}}}_2^2
			\leq \frac{c_0}{\sigma_\eta (n-1) \varepsilon^2} 
			\sum_{\set{x,y} \in \E_{n,\eps}} \eta_\eps(\abs{x-y}) \left({u(x) - u(y)}\right)^2
		\end{equation}
		holds for all $u \in \ell^2({\G}_{n,\eps})$ has probability at least $1 - c_1 n e^{-c_2 n \eps^d}$. Here
		\begin{equation}\label{eq:weighted_mean_degree}
			(u)_{\deg_{n,\eps}} := \frac{\sum_{x \in \V_n} u(x) \deg_{n,\eps}(x)}{\sum_{x \in \V_n} \deg_{n,\eps}(x)}
		\end{equation}
		is the degree-weighted mean of $u$.
	\end{proposition}
	\begin{remark}
		We remark that in the definition of $\norm{\cdot}_2$ we follow a slightly different convention than the paper \cite{Bungert:ConvergenceRatesPoisson2024}, in that we do not scale it with $\frac1n$. Therefore, the original proposition has an additional factor $\frac1n$ on the right hand side of \eqref{eq:discrete_poincare}.
		
		Moreover, we note that 
		\begin{equation}
			\frac1{\sigma_\eta n \varepsilon^2} 
			\sum_{\set{x,y} \in \E_{n,\eps}} \eta_\eps(\abs{x-y}) \left({u(x) - u(y)}\right)^2
			= (u, -\tilde\Delta_{\G_{n, \eps}} u),
		\end{equation}
		where $-\tilde \Delta_{\G_{n, \eps}}$ is the graph Laplacian with respect to the coupling constants
		\begin{equation}	
			\tilde\J_{n,\eps}(x, y) := 
			\frac{\eta_\eps(\abs{x - y})}{\sigma_\eta n \varepsilon^2}, \quad x,y \in \V_n.
		\end{equation}
		Then \eqref{eq:discrete_poincare} gives a lower bound on the first non-zero eigenvalue of $-\tilde\Delta_{\G_{n, \eps}}$, which is one justification for the name ``Poincaré inequality'' of \zcref{thm:discrete_poincare}. We use this estimate in \zcref{sec:random_graphs} to prove bounds on the effective resistance of $\G_{n, \eps}$.
	\end{remark}

	\section{Order in deterministic graphs}
	\label{sec:finite_graphs}
	
	Let $G=(V,E,J)$ be any finite, undirected graph. In this section we  derive a lower bound on the two-point correlation function $\cos\theta_x-\theta_y$ in the $XY$ measure on $G$ defined in \eqref{eq:def_$XY$_on_G} in terms of the effective resistance of the graph.
	Recall the definition of the graph Laplacian \eqref{eq:def_graph_laplacian} \eqref{eq:def_graph_laplacian1}
	\begin{equation}
		(v, -\Delta_G v) = \sum_{\set{i,j} \in E} J(i,j) \left({v(i) - v(j)}\right)^2, \qquad v \in \ell^2(G).
	\end{equation}
	
	\begin{definition}[Effective resistance]\label{def:effres} 
		The \emph{effective resistance} of $G$ is given by
		\begin{equation}\label{eq:def_effective_resistance}
			\mathbf{R}_{G}:= \sup_{u,v\in V} \left(e_u - e_v, \left(-\Delta_{G}\right)^{-1} \left(e_u - e_v\right) \right),
		\end{equation}	
		where for $x \in V$ we define $e_x \in \ell^2(G)$ by $e_x(v) = 1$ if $v = x$ and $0$ otherwise.
		$\mathbf{R}_{G}$ is well defined, because $e_{u}-e_{v}\in (\ker (-\Delta_{G}))^{\perp}$ (cf. \eqref{eq:def_graph_laplacian1}).
	\end{definition}
	
	\begin{remark}
		The effective resistance between two points $x,y \in V$ has a close connection with the hitting time of $y$ by a simple random walk on $G$ started in $x$, cf. \cite{Lyons:ProbabilityTreesNetworks2016,Doyle:RandomWalksElectric1984}.
	\end{remark}
	Recall  the definition of  algebraic irreducibility 
\eqref{eq:defirr}.	The main result of this section is the following theorem.
	
	%%%%%%%%%%%%%%%%%%%%%%%%%%%%%%%%%%%%%%%%%%%%%%%%%%%%%%%%%%%%%%%%%%%%%%%%%%%%%%%%%%%
	\begin{theorem}\label{thm:expectation_bound_fixed_graph}
	 Let $G=(V, E, J)$ be a finite and connected graph and 	 $0 < \kappa <  \frac{\pi}{\irr(G) }$. Let
		$J_0 := \min\set{J(i,j) \mid \set{i,j} \in E} > 0$ as in \eqref{def:J0}. 
		Then there exist a positive constant $c_0$ independent of $G$ and $\kappa$, such that for all $\beta > 0$ and all $x, y \in V$.
		\begin{equation}	
			\sprod{\cos(\theta_x - \theta_y)}{G;\beta} 
			\geq 1
			- \tfrac{\mathbf{R}_{G}}{2 \beta \cos \kappa}	- c_0 \abs{E} \, e^{-\frac{\beta J_0 \kappa^2}{\pi}}.
			\end{equation}
	\end{theorem}
%%%%%%%%%%%%%%%%%%%%%%%%%%%%%%%%%%%%%%%%%%%%%%%%%%%%%%%%%%%%%%%%%%%%%%%%%%%%%%%%%%%%%%%%%%%%%%%%%%%%%5

	To prove the result we need the following definition.	
%%%%%%%%%%%%%%%%%%%%%%%%%%%%%%%%%%%%%%%%%%%%%%%%%%%%%%%%%%%%%%%%%%%%%%%%%%%	
	\begin{definition}
		Let $\theta \in \R^V$. For $u,v \in V$ we define the \emph{gradient}
		\begin{equation}
			\nabla_{uv} \theta := \theta_v - \theta_u
		\end{equation}
		and the \emph{periodic gradient}
		\begin{equation}
			\nabla^{\per}_{uv} \theta := \nabla_{uv} \theta + 2 \pi k_{uv} \quad \text{where $k_{uv} \in \Z$ such that} \quad \nabla_{uv} \theta + 2 \pi k_{uv} \in [-\pi, \pi).
		\end{equation}
		Further, for $\kappa > 0$, define
		\begin{equation}
			\mathcal{G}_\kappa := \set{\omega \in \R^E \mid \abs{\omega(e)} < \kappa \text{ for all } e \in E},
		\end{equation}
		which we call the $\kappa$-good edge configurations. 
	\end{definition}
%%%%%%%%%%%%%%%%%%%%%%%%%%%%%%%%%%%%%%%%%%%%%%%%%%%%%%%%%%%%%%%%%
	\begin{remark}\label{remark:symmetry_gradient}
		While $\nabla_{uv} \theta=-\nabla_{vu} \theta$ holds for all $u,v\in V$, we have $\nabla^{\per}_{uv} \theta =-\nabla^{\per}_{vu} \theta $
		only when $\nabla^{\per}_{uv} \theta\in (-\pi,\pi )$, while $\nabla^{\per}_{uv} \theta =\nabla^{\per}_{vu} \theta$ when $\nabla^{\per}_{uv} \theta=-\pi$.
	Since $|\nabla^{\per}_{uv} \theta|=|\nabla^{\per}_{vu} \theta|$, 
		by a small abuse of notation, we will write $\gradper \theta \in \mathcal{G}_\kappa$ if $\abs{\gradper_{uv}\theta } \leq \kappa$ for all
		$\set{u,v} \in E$. The same applies when writing $\nabla\theta\in \mathcal{G}_\kappa$.
	\end{remark}
%%%%%%%%%%%%%%%%%%%%%%%%%%%%%%%%%%%%%%%%%%%%%%%%%%%%%%%%%%
	\begin{proof}[Proof of \zcref{thm:expectation_bound_fixed_graph}]
		We argue
	\begin{align}
		\sprod{1-\cos(\theta_x - \theta_y)}{G;\beta} &=\sprod{1 - \cos(\theta_x - \theta_y) \mid \gradper \theta \in \mathcal{G}_\kappa}{G;\beta}\ 
		\mu_{G;\beta}\big(\gradper \theta \in \mathcal{G}_\kappa\big)\nonumber\\
		&+
		\sprod{1 - \cos(\theta_x - \theta_y) \mid \gradper \theta \notin \mathcal{G}_\kappa}{G;\beta}\ 
		\mu_{G;\beta}\big(\gradper \theta \notin \mathcal{G}_\kappa\big)\nonumber\\
		&\leq \sprod{1 - \cos(\theta_x - \theta_y) \mid \gradper \theta \in \mathcal{G}_\kappa}{G;\beta}+2\mu_{G;\beta}\big(\gradper \theta \notin \mathcal{G}_\kappa\big).
	\end{align}
	The second term is exponentially small by known results 
	on nearest neighbor alignement (cf. \zcref{th:NNalign}). 
	The rest of the proof is devoted to show the bound
	\begin{equation}	
			\sprod{1 - \cos(\theta_x - \theta_y) \mid \gradper \theta \in \mathcal{G}_\kappa}{G;\beta} 
			\leq \frac{\left(e_x - e_y, \left(-\Delta_{G}\right)^{-1} \left(e_x - e_y\right) \right)}{2 \beta \cos \kappa}
			\leq \frac{\mathbf{R}_{G}}{2 \beta \cos \kappa}.
	\end{equation}
	Our strategy is to apply the Brascamp--Lieb inequality to the conditional measure. However, since the conditioning event
	is non-convex, we must first replace it by a constraint on the real gradient $\grad \theta$. This will require to change the measure,
	which we will do by a change of variables. The construction is inspired by the work of Newman and Wu \cite{Newman:GaussianFluctuationsClassical2018} where they use a coupling argument instead.
 	Fix an arbitrary point $x_* \in V$ and define
		\begin{equation}\label{def:K}
			\mathcal K := \set{k \in  \Z^V \colon k_{x_*} = 0 \text{ and } k_i - k_j \in \set{-1, 0, 1} \text{ if } \set{i,j} \in E}.
		\end{equation}
	and 
	\begin{equation}\label{def:ai}
		a_i := \max_{k \in \mathcal K} 2\pi \abs{k_i} + \pi,\quad i\in V.
	\end{equation}
		In \zcref{thm:change_of_variables} we prove the following change of variables formula
		\begin{equation}\label{eq:changecoord}
				\int_{[-\pi,\pi  )^{V}} f(\theta)\ \I_{\mathcal{G}_\kappa}(\gradper \theta) \d \theta 
				= \int_{\times_{i\in V}[-a_{i},a_{i})} f(\vartheta)\  \I_{\mathcal{G}_\kappa}(\grad \vartheta) \d \vartheta 
			\end{equation}
			for any $2\pi$-periodic function $f$, and $0 < \kappa < \frac{2\pi}{\irr(G) }$.
	Hence
	\begin{equation}
		\sprod{1 - \cos(\theta_x - \theta_y) \mid \gradper \theta \in \G_\kappa}{G; \beta}
		= \frac{\int_{\times_{i\in V}[-a_{i},a_{i})}  (1-\cos(\vartheta_x - \vartheta_y)) e^{-\beta H(\vartheta)} \I_{\mathcal{G}_\kappa}(\grad \vartheta)\d \vartheta}{\int_{\times_{i\in V}[-a_{i},a_{i})} e^{-\beta H(\vartheta)} \I_{\mathcal{G}_\kappa}(\grad \vartheta)\d \vartheta},
	\end{equation}
	where
	\begin{equation}
		H(\vartheta)
		:= -\sum_{\set{i,j} \in E} J(i,j) \cos(\theta_i - \theta_j).
	\end{equation}
	We show that we can approximate the measure by a log-convex one defined on $\R^V$, which will allow the application of the Brascamp--Lieb inequality.
	Namely, define for $m,p > 0$
	\begin{equation}
		H_{m,p}(\vartheta)
		:= -\sum_{\set{i,j} \in E} J(i,j) \left(\cos(\vartheta_i - \vartheta_j) - \left(\tfrac{\vartheta_i -\vartheta_j}{\kappa}\right)^{2p}\right) + \sum_{i \in V} \left(\tfrac{m \cos \kappa}2 \vartheta_i^2 + \left(\tfrac{\vartheta_i}{a_i}\right)^{2p}\right),
	\end{equation}
	where the $a_k$ are given in \eqref{def:ai}. 
		The $p$ terms approximate the constraints on $\vartheta_{V}$ since 
	\begin{equation}
		e^{-\beta\left(\frac{\vartheta_i - \vartheta_j}{\kappa}\right)^{2p}} \to \I_{\abs{\vartheta_i - \vartheta_j} \leq \kappa} 
		\quad\text{and}\quad
		e^{-\beta\left(\frac{\vartheta_i}{a_i}\right)^{2p}} \to \I_{\abs{\vartheta_i} \leq a_i}
	\end{equation}
	pointwise almost everywhere as $p \to \infty$. The $m$ term acts like a mass making the Hessian strictly convex.
	Therefore, 
	\begin{equation}
		\sprod{1 - \cos(\theta_x - \theta_y) \mid \gradper \theta \in \G_\kappa}{G;\beta}
		= \lim_{m \to 0} \lim_{p \to \infty} \frac{\int_{\R^V}  (1-\cos(\vartheta_x - \vartheta_y)) e^{-\beta H_{m,p}(\vartheta)}\d \vartheta}{\int_{\R^V} e^{-\beta H_{m,p}(\vartheta)}\d \vartheta},
	\end{equation}
	where we use dominated convergence to exchange the pointwise limits with the integral. Let $H_{m,p}''$ be the Hessian of $H_{m,p}$. For $w, \vartheta \in \R^V$ and $p^2 \geq \kappa^2(1+\cos \kappa)$, we obtain
	\begin{equation}\begin{aligned}
		(w, H_{m,p}''(\vartheta) w)
		&= \sum_{\set{i,j} \in E} J(i,j) \left(\cos(\vartheta_i - \vartheta_j) + \tfrac{2p(2p-1)}{\kappa^{2p}}\left(\vartheta_i -\vartheta_j\right)^{2p-2}\right)(w_i - w_j)^2 +\\ 
		&\qquad	+ \sum_{i \in V} \left(m \cos \kappa + \tfrac{2p(2p-1)}{a_i^{2p}}\left(\vartheta_i\right)^{2p-2}\right) w_i^2 \\
		&\geq \cos \kappa \sum_{\set{i,j} \in E} J(i,j)(w_i-w_j)^2
		+ m \cos \kappa \sum_{i \in V} w_i^2 =  \cos\kappa\  \big(w, \left(-\Delta_{G} + m \Id \right) w\big),
	\end{aligned}\end{equation}
	where we used $\frac{\left(\vartheta_i -\vartheta_j\right)^{2p-2}}{\kappa^{2p-2}}\! \geq \! 1$ if $\abs{\vartheta_i -\vartheta_j} \geq \kappa$, while $\cos(\vartheta_i -\vartheta_j) \geq \cos\kappa$ if $\abs{\vartheta_i -\vartheta_j} < \kappa$. Hence, as quadratic forms,
	\begin{equation}
		H_{m,p}''(\vartheta) \geq \cos\kappa \ (-\Delta_{G} + m \Id) > 0.
	\end{equation}
	Applying the Brascamp--Lieb inequality \cite{Brascamp:ExtensionsBrunnMinkowskiPrekopaLeindler1976,Helffer:SemiclassicalAnalysisWitten2002,Johnsen:SpectralPropertiesWittenLaplacians2000} to $f(\vartheta) := \vartheta_x - \vartheta_y$ we obtain
	\begin{equation}
		\begin{aligned}
			\frac{\int_{\R^V}  (1-\cos(\vartheta_x - \vartheta_y)) e^{-\beta H_{m,p}(\vartheta)}\d \vartheta}{\int_{\R^V}  e^{-\beta H_{m,p}(\vartheta)}\d \vartheta}  
		&	\leq \frac12 \frac{\int_{\R^V} \left(\vartheta_x - \vartheta_y\right)^2 e^{-\beta H_{m,p}(\vartheta)} \d \vartheta}{\int_{\R^V} e^{-\beta H_{m,p}(\vartheta)} \d \vartheta} \nonumber\\
			&\leq \frac{1}{2 \beta \cos\kappa} \left(e_x - e_y, \left(-\Delta_{G} + m \Id\right)^{-1}(e_x - e_y)\right).\label{eq:apply_BL}
		\end{aligned}
	\end{equation}
	The addition of the small mass $m$ makes the operator $(-\Delta_{G} + m \Id)$ positive and, since it is just a matrix, invertible.
	To remove the mass, note that $-\Delta_{G} + m \Id> -\Delta_{G} >0$ on  $(\ker ( -\Delta_{G}))^{\perp}$.
	Therefore,
	\begin{equation}
		\left(e_x - e_y, \left(-\Delta_{G} + m \Id\right)^{-1}(e_x - e_y)\right)
		\leq \left(e_x - e_y, \left(-\Delta_{G}\right)^{-1}(e_x - e_y)\right)
		\leq \mathbf{R}_G,
	\end{equation}
	concluding the proof.
	\end{proof}
%%%%%%%%%%%%%%%%%%%%%%%%%%%%%%%%%%%%%%%%%%%%%%%%%%%%%%%%%%%%%%%%%%%%%%%%%%%%%%%%%%%%%%%%%%%%%%

	The rest of this section is devoted to prove the change of coordinates formula \eqref{eq:changecoord}. 

	\begin{proposition}\label{thm:change_of_variables}
		Let $G = (V, E)$ be a finite and connected graph and let $0 < \kappa < \frac{2\pi}{\irr(G)}$. For any $f: \R^V \to \R$ which is $2\pi$-periodic it holds 
		\begin{equation}
			\int_{[-\pi, \pi)^V} f(\theta)\  \I_{\mathcal{G}_\kappa}(\gradper \theta) \d \theta 
			= \int_{\times_{i\in V}[-a_{i},a_{i})}  f(\vartheta)\  \I_{\mathcal{G}_\kappa}(\grad \vartheta) \d \vartheta .
		\end{equation}
	\end{proposition}
%%%%%%%%%%%%%%%%%%%%%%%%%%%%%%%%%%%%%%%%%%%%%%%%%%%%%%%%%%%%%%%%%%%%%%%%%%%%%%%%%%%

	\begin{proof} Recall the definition of $\mathcal{K}$ \eqref{def:K}.
	We introduce the map $\Psi \colon [-\pi, \pi)^V \to  2\pi\mathcal K$ defined by
	\begin{equation}
		\Psi(\theta)_x := \begin{cases}
			\sum_{k=0}^{n-1} (\nabla^{\per}_{z_k z_{k+1}} \theta) - \left(\theta_x - \theta_{x_*}\right)	& \text{if } \gradper \theta \in \mathcal{G}_\kappa, \\
			0 &\text{else}.
		\end{cases}
	\end{equation}
	Here, $z_0\cdots z_n$, $z_0 = x_*$, $z_n = x$ is any path from $x_*$ to $x$ in $G$. Since $G$ is connected, such a path always exists. In \zcref{thm:Psi_well_defined} we prove that this definition is well defined, i.e. independent of the chosen path. 
	Moreover, we prove in \zcref{thm:gradper_grad_equivalence} that for $\theta \in [-\pi,\pi)^V$ and $k \in \mathcal K$
	\begin{equation}
		\gradper \theta \in \mathcal{G}_\kappa \text{ and } \Psi(\theta) = 2\pi k 
		\quad \text{if and only if} \quad
		\grad(\theta + 2\pi k) \in \mathcal{G}_\kappa
	\end{equation}
	whenever $0 < \kappa < \frac{2\pi}{\irr(G)}$.
	Hence,
	\begin{equation}\begin{aligned}
			&\phntm\int_{[-\pi, \pi)^V} f(\theta)\  \I_{\mathcal{G}_\kappa}(\gradper \theta) \d \theta 
		 	= \sum_{k \in \mathcal K}\int_{[-\pi, \pi)^V} f(\theta) \ \I_{\mathcal{G}_\kappa}(\gradper \theta) \I_{\Psi(\theta) = 2\pi k} \d \theta\\
			&\qquad = \sum_{k \in \mathcal K}\int_{[-\pi, \pi)^V} f(\theta)\  \I_{\mathcal{G}_\kappa}(\grad \left(\theta + 2\pi k\right)) \d \theta
			= \sum_{k \in \mathcal K}\int_{[-\pi, \pi)^V +2\pi  k} f(\vartheta-2\pi k)\ \I_{\mathcal{G}_\kappa}(\grad \vartheta)  \d \vartheta\\
			&\qquad =\sum_{k \in \mathcal K}\int_{[-\pi, \pi)^V +2\pi  k} f(\vartheta)\  \I_{\mathcal{G}_\kappa}(\grad \vartheta)  \d \vartheta
			= \int_{\times_{i\in V}[-a_{i},a_{i})} f(\vartheta)\  \I_{\mathcal{G}_\kappa}(\grad \vartheta) \d \vartheta,
	\end{aligned}\end{equation}
	where in the last line we used the $2\pi$-periodicity of $f$ and the identity		
	\[
		\bigcup_{k\in \mathcal K} [-\pi, \pi)^V +2\pi  k\ \cap\ \{\nabla \vartheta \in \mathcal{G}_\kappa \}=
		\times_{i\in V}[-a_{i},a_{i})\cap\ \{\nabla \vartheta \in \mathcal{G}_\kappa \}.
	\]
	To prove this identity, note that the inclusion from left to right follows from the definition of $a_{i}$. For the other inclusion, let
	$ \vartheta\in  \times_{i\in V}[-a_{i},a_{i})$ with $ \nabla\vartheta \in \mathcal{G}_\kappa. $ Then  $\vartheta= \theta +2\pi k$ with $k\in \Z^{V}$ and $\theta \in  [-\pi, \pi)^V$.  It remains to show that $k\in  \mathcal K$ holds.
	 By definition $a_{x_{*}}=\pi$ and hence $k_{x_{*}}=0$. Moreover, if $\{i,j \}\in E$, we argue
	\[
	2\pi |k_{i}-k_{j}|= |\vartheta_{i}-\vartheta_{j}- (\theta_{i}-\theta_{j}) |\leq \kappa +2\pi. 
	\]
	Since $\irr(G)\geq 3$ (cf. \zcref{rem:irrG}) we have $\kappa <\pi $ and hence $k_{i}-k_{j}\in \{-1,0,1 \}$.
		\end{proof}
	
	The next lemma shows that the function $\Psi$ is well defined.

%%%%%%%%%%%%%%%%%%%%%%%%%%%%%%%%%%%%%%%%%%%%%%%%%%%%%%%%%%%%%%%%%%%%%%%%%%%%%%	
	\begin{lemma}\label{thm:Psi_well_defined}
		Let $G=(V, E)$ be a finite and connected graph and $0 < \kappa < \frac{2\pi}{\irr(G)}$. Then $\Psi$ is well defined, i.e. for all $x \in V$ its definition does not depend on the choice of the path from $x_*$ to $x$, and $\Psi(\theta) \in \mathcal K$ for all $\theta \in [-\pi,\pi)^V$. 
	\end{lemma}
%%%%%%%%%%%%%%%%%%%%%%%%%%%%%%%%%%%%%%%%%%%
	\begin{proof}
		If $\gradper \theta \notin \mathcal{G}_\kappa$, then $\Psi(\theta) = 0 \in \mathcal K$. Thus, from now on we assume that $\gradper \theta \in \mathcal{G}_\kappa$.
			Let $\gamma=y_0 \cdots y_n$, and $\gamma'=z_0 \cdots z_m$ with $y_0 = z_0 = x_*$ and $y_n =z_m = x$ be two different paths from $x_*$ to $x$ in $G$.
		For $\Psi$ to be well defined, we need to show that
		\begin{equation}\label{eq:equality_Psi_on_paths}
			\sum_{k=0}^{n-1} \gradper_{y_{k} y_{k+1}} \theta = \sum_{k=0}^{m-1} \gradper_{z_{k} z_{k+1}} \theta.
		\end{equation}
		Recall that $\irr(G) \geq 3$ and therefore $\kappa < \pi$. Since $\gradper \theta \in \mathcal G_\kappa$, we  have $\gradper_{ij} \theta = - \gradper_{ji} \theta$ as was noted in \zcref{remark:symmetry_gradient}.
		Hence, \eqref{eq:equality_Psi_on_paths} is equivalent to
		\begin{equation}\label{eq:gradper_different_paths_condition}
			\sum_{k=0}^{n-1} \gradper_{y_{k} y_{k+1}} \theta + \sum_{k=0}^{m-1} \gradper_{z_{k+1} z_{k}} \theta = 0.
		\end{equation}
		Let $\vec G = (V, \vec E)$ be any orientation of $G$.
		We can rewrite
		\begin{equation}
			\begin{aligned}
				\sum_{k=0}^{n-1} \gradper_{y_k y_{k+1}} \theta
				&= \sum_{k=0}^{n-1} \left(\gradper_{y_k y_{k+1}} \theta\right) \sum_{\vec e \in \vec E} \left(\delta_{\vec e, (y_k, y_{k+1})} + \delta_{\vec e, (y_{k+1}, y_k)}\right) \\
				&= \sum_{\vec e \in \vec E} \left(\gradper_{\vec e} \theta \right) \sum_{k=0}^{n-1} \left(\delta_{\vec e, (y_k, y_{k+1})} - \delta_{\vec e, (y_{k+1}, y_k)}\right),
			\end{aligned}
		\end{equation}
		where we used $\gradper_{y_{k+1} y_{k}} = -\gradper_{y_{k} y_{k+1}}$ and either $(y_k, y_{k+1}) \in \vec E$ or $(y_{k+1}, y_{k}) \in \vec E$. Similarly,
		\begin{equation}
				\sum_{k=0}^{m-1} \gradper_{z_{k+1} z_{k}} \theta
				= \sum_{\vec e \in \vec E} \left(\gradper_{\vec e} \theta \right) \sum_{k=0}^{n-1} \left(\delta_{\vec e, (z_{k+1}, z_k)} - \delta_{\vec e, (z_k, z_{k+1})}\right).
		\end{equation}
		Defining $\vec f \in \mathbf E(\vec G)$ by
		\begin{equation}
			\vec f(\vec e) := \sum_{k=0}^{n-1} \left( \delta_{\vec e, (y_{k}, y_{k+1})} - \delta_{\vec e, (y_{k+1}, y_{k})} \right) 
			+  \sum_{k=0}^{m-1} \left( \delta_{\vec e, (z_{k+1}, z_{k})} - \delta_{\vec e, (z_{k}, z_{k+1})} \right), \quad \vec e \in \vec E,
		\end{equation}
		we obtain
		\begin{equation}\label{eq:sum_gradper_as_dot_product}
			\sum_{k=0}^{n-1} \gradper_{y_{k} y_{k+1}} \theta + \sum_{k=0}^{m-1} \gradper_{z_{k+1} z_{k}} \theta
			= \sum_{\vec e \in \vec E} \vec f(\vec e) \gradper_{\vec e} \theta.
		\end{equation}
		Further,  we have  $\vec f \in \mathbf{C}(\vec G)$ since each vertex has the same number of incoming and outgoing edges in $\vec f$  (cf. also the argument in the proof of \zcref{thm:lifting_exists}). 
			Let now $f_1, \ldots, f_N \in \mathbf C(G)$ be a basis of $\mathbf C(G)$ such that
		\begin{equation}
			\abs{f_k} < \frac{2\pi}{\kappa} \quad \text{for all } k=1, \ldots, N.
		\end{equation}
		Such a basis exists, because  $\irr(G) < \frac{2\pi}{\kappa}$ by assumption.
		Let $\vec f_i \in \mathbf C(\vec G)$, $i \in \set{1, \dots, N}$, be a lifting of $f_i$. By \zcref{thm:lifting_basis} the $\vec f_1, \dots, \vec f_N$ form a basis of $\mathbf C(\vec G)$ and therefore
		\begin{equation}
			\vec f = \sum_{k=1}^N \lambda_k \vec f_k
		\end{equation}
		for some $\lambda_1, \dots, \lambda_N \in \Q$. Inserting into \eqref{eq:sum_gradper_as_dot_product} we obtain
		\begin{equation}\label{eq:decompose_f_gradper_by_basis}
			\sum_{k=0}^{n-1} \gradper_{y_{k} y_{k+1}} \theta + \sum_{k=0}^{m-1} \gradper_{z_{k+1} z_{k}} \theta
			= \sum_{k=1}^N \lambda_k \left(\sum_{\vec e \in \vec E} \vec f_k(\vec e) \gradper_{\vec e} \theta\right).
		\end{equation}
		For any edge $(i,j) \in \vec E$ we have $\gradper_{ij} \theta = \theta_j - \theta_i + 2\pi k_{ij}$, with $k_{ij} \in \Z$. Thus,
		\begin{equation}
			\sum_{\vec e \in \vec E} \vec f_k(\vec e) \gradper_{\vec e} \theta
			= \sum_{(i,j) \in \vec E} \vec f_k((i,j)) \theta_j - \sum_{(i,j) \in \vec E} \vec f_k((i,j)) \theta_i + 2\pi \sum_{(i,j) \in \vec E} \vec f_k((i,j)) k_{ij}.
		\end{equation}
		Note that
		\begin{equation}
			\vec E = \bigcup_{v \in V} \delta^+(v) = \bigcup_{v \in V} \delta^-(v).
		\end{equation}
		Hence
		\begin{equation}
			\sum_{(i,j) \in \vec E} \vec f_k((i,j)) \theta_j
			= \sum_{j \in V}  \theta_j \sum_{(i,j) \in \delta^-(j)} \vec f_k((i,j))
			= \sum_{j \in V}  \theta_j \sum_{(j,i) \in \delta^+(j)} \vec f_k((j,i))
			= \sum_{(i,j) \in \vec E} \vec f_k((i,j)) \theta_i,
		\end{equation}
		where we used that $\vec f_k \in \mathbf C(\vec G)$ (and thus satisfies \eqref{eq:directed_cycle_contraint}) for the second equality. Thus, using that the $\vec f_k$ are liftings of $f_k$ and therefore $\vec f_k(\vec E) \subset \set{-1, 0, 1}$,
		\begin{equation}
			\sum_{e \in \vec E} \vec f_k(e) \gradper_e \theta = 2\pi \sum_{(i,j) \in \vec E} \vec f_k(e) k_{ij} \in 2\pi \Z.
		\end{equation}
		On the other hand, we assumed $\gradper\theta \in \mathcal{G}_\kappa$. Therefore,
		\begin{equation}
			\abs{\sum_{e \in \vec E} \vec f_k(e) \gradper_e \theta}
			\leq \kappa \sum_{e \in \vec E} \abs{\vec f_k(e)}
			= \kappa \abs{f_k} < 2\pi.
		\end{equation}
		Hence,
		\begin{equation}
			\sum_{e \in \vec E} \vec f_k(e) \gradper_e \theta = 0.
		\end{equation}
		Together with \eqref{eq:decompose_f_gradper_by_basis} this proves \eqref{eq:gradper_different_paths_condition}.
		
		It remains to show $\Psi(\theta) \in \mathcal K$. By choosing the empty path, one sees that $\Psi(\theta)_{x_*} = 0$. Let $\set{u,v} \in E$. Let $z_0 \cdots z_n$, $z_0 = x_*$, $z_n = u$ be a path from $x_*$ to $u$ in $G$. Then $z_0 \cdots z_n \,v$ is a path from $x_*$ to $v$. Thus,
		\begin{equation}
			\begin{aligned}
				\Psi(\theta)_u - \Psi(\theta)_v
				&= \sum_{k=0}^{n-1} \gradper_{z_{k} z_{k+1}} \theta - \sum_{k=0}^{n-1} \gradper_{z_{k} z_{k+1}} - \gradper_{uv} \theta - (\theta_u - \theta_{x_*}) + (\theta_v - \theta_{x_*}) \\
				&= -\gradper_{uv} \theta + \theta_v - \theta_u 
				= -2 \pi k_{uv} \in \set{-2\pi, 0, 2\pi},
			\end{aligned}
		\end{equation}
		where we used that $\theta \in [-\pi, \pi)^V$ and thus $\theta_v - \theta_u \in (-2\pi, 2\pi)$. Hence, $\Psi$ is well defined.
	\end{proof}

	The last lemma of this section relates the $\kappa$-goodness of the periodic gradient to the $\kappa$-goodness of the gradient of a translated field.
	\begin{lemma}\label{thm:gradper_grad_equivalence}
		Let $G = (V, E)$ be a finite and connected graph and $0 < \kappa < \frac{2\pi}{\irr(G)}$. Let $\theta \in [-\pi, \pi)^V$ and $k \in \mathcal K$. 
		Then 
		\begin{equation}
			\gradper \theta \in \mathcal{G}_\kappa \text{ and } \Psi(\theta) =2\pi  k 
			\quad \text{if and only if} \quad
			\grad(\theta +2\pi  k) \in \mathcal{G}_\kappa.
		\end{equation}
	\end{lemma}
%%%%%%%%%%%%%%%%%%%%%%%%%%%%%%%%%%%%%%%%%%%%%%%%%%%%%%%%%%%%%%%%%%%%%%%%%%%%%%%%%%%%%%%%%%%
	\begin{proof}
		For one direction, let $\set{u,v} \in E$ and assume that $\gradper \theta \in \mathcal{G}_\kappa$ and $\Psi(\theta) =2\pi  k$. 
		Let $z_0 \cdots z_n$, $z_0 = x_*$, $z_n = u$ be a path from $x_*$ to $u$ in $G$. Then $z_0 \cdots z_n \,v$ is a path from $x_*$ to $v$. Hence
		\begin{equation}
			\begin{aligned}
				\gradper_{uv} \theta
				&= \left(\sum_{l=0}^{n-1} \gradper_{z_l, z_{l+1}} \theta + \gradper_{uv} \theta - \theta_v + \theta_{x_*}\right) - \left(\sum_{l=0}^{n-1} \gradper_{z_l, z_{l+1}} \theta - \theta_u + \theta_{x_*}\right) + \theta_v - \theta_u \\
				&= \Psi(\theta)_v - \Psi(\theta)_u + \theta_v - \theta_u
				= \theta_v + 2\pi k_v - (\theta_u +2\pi  k_u )
				= \grad_{uv} \left(\theta + 2\pi k\right).
			\end{aligned}
		\end{equation}
		Thus $\grad \left(\theta +2\pi  k\right) \in \mathcal{G}_\kappa$.
			
		For the other direction, assume that $\grad \left(\theta +2\pi  k\right) \in \mathcal{G}_\kappa$.
		Let $\set{u,v} \in E$ and $n_{uv} \in \Z$ such that $\gradper_{uv} \theta = \theta_v - \theta_u +2\pi  n_{uv}$.
		Then
		\begin{equation}
		2\pi \abs{n_{uv} -  k_v + k_u} =	\abs{\gradper_{uv} \theta - \grad_{uv} (\theta +2\pi  k)} 
			\leq \pi + \kappa.
		\end{equation}
		On the other hand, $k \in \Z^V$, and thus $n_{uv} - k_v + k_u \in  \Z$. Since $\kappa < \pi$, the only option is that $n_{uv} = k_v - k_u$ and therefore
		\begin{equation}\label{eq:gradper_equals_grad_plus_k}
			\gradper \theta = \grad (\theta + 2\pi k) \in \mathcal{G}_\kappa.
		\end{equation}
		For $x \in V$ let $z_0 \cdots z_n$, $z_0 = x_*$, $z_n = x$ be a path from $x_*$ to $x$ in $G$. Then, using \eqref{eq:gradper_equals_grad_plus_k},
		\begin{equation}
				\Psi(\theta)_x
				= \sum_{l=0}^{n-1} \nabla_{z_l, z_{l+1}} (\theta +2\pi k) - (\theta_x - \theta_{x_*}) 
				= \theta_{z_n} +2\pi  k_{z_n} - \theta_{z_0} -2\pi k_{z_0} - (\theta_x - \theta_{x_*}) 
				=2\pi k_x,
		\end{equation}
		concluding the proof.
	\end{proof}
	%%%%%%%%%%%%%%%%%%%%%%%%%%%%%%%%%%%%%%%%%%%%%%%%%%%%%%%%%%%%%%%%%%%%
	\section{Order in random geometric graphs}
	\label{sec:order_rand_geometric_graphs}
	
	We use \zcref{thm:expectation_bound_fixed_graph}, established in the previous section, to prove our main \zcref{thm:main}. For this, we consider realizations $\omega \in \Xi$ of the random environment, on which we have uniform control for $\irr(\G_{n, \eps}(\omega))$ and $\mathbf R_{\G_{n, \eps}(\omega)}$. In \zcref{sec:_irr_random_graphs} we show that there does indeed exist a subset $\mathcal A_1 \subset \Xi$ with very large probability, on which the algebraic irreducibility of the random graph is bounded by some constant $K$ depending only on $\Omega$, $\eta$ and $\rho$. Then, in \zcref{sec:random_graphs}, we prove that there exists a second set $\mathcal A_2 \subset \Xi$ which also has a very large probability and on which the random effective resistance is bounded by $R\, \eps^{-2}$ for a  constant $R$ depending only on $\Omega$, $\eta$ and $\rho$. These two sets are used in \zcref{sec:random_graphs} to prove the main theorem.

%%%%%%%%%%%%%%%%%%%%%%%%%%%%%%%%%%%%%%%%%%%%%%%%%%%%%%%%%%%%%%%%%%%%%%%%%
	\subsection{Algebraic irreducibility of random geometric graphs}
	\label{sec:_irr_random_graphs}
	
	In this section, we will show that the algebraic irreducibility of random geometric graphs can be controlled uniformly in $n$ and $\varepsilon$.
	\begin{theorem}\label{thm:irr_random_graph}
		Assume that $\Omega$, $\eta$ and $\rho$ satisfy \zcref{ass:Omega,ass:eta, ass:rho}.	
		Additionally, assume that $\Omega$ is simply connected. 
		
		Then there exist positive constants 
		$c_1 \equiv c_1(\Omega)$, 
		$c_ 2 \equiv c_2(\rho_{\min}, \Omega)$, 
		$\eps_0 \equiv \eps_0(\Omega)$ and
		$K \equiv K(\Omega, \supp \eta)$, such that for any $\eps \in (0,\eps_0]$ 
		\begin{equation}\label{eq:prob_connected_and_irr_bounded}
			\P\left({\G}_{n,\eps} \text{ is connected and } \irr\left({{\G}}_{n,\eps}\right)\leq K\right) \geq 1 - c_1 n e^{-c_2 n \eps^d}.
		\end{equation}
	\end{theorem}

	We postpone the proof of this theorem to the end of the section and introduce first the two main technical results that we will need for the argument. The first result (\zcref{thm:project_Omega_to_Lambda_h}) shows that $\Omega $ can be projected  onto the lattice $\Lambda_h := h\Z^d \cap \Omega$ defined in \eqref{def:Lambda_h} without moving  points too far (for $h$ small enough).  The second result (\zcref{thm:existence_discretization}) shows that continuous curves in $\Omega $ induce paths in $\Lambda_h$ without losing much. Recall that we will take $h$ proportional to $\varepsilon$ in the subsequent applications.
	
	\begin{definition}
		Let $h > 0$. We define the map $\pi_h \colon \Omega \to \Lambda_h$ by
		\begin{equation}\label{eq:def_pi_h}
			\pi_h(x) := 
			\Phi_h\left(\set{v \in \Lambda_h \colon \abs{x - v} \leq \abs{x-w} \text{ for all } w \in \Lambda_h} \right),
		\end{equation}
		where $\Phi_h \colon 2^{\Lambda_h} \to \Lambda_h$ is some rule selecting a single element from a subset of $\Lambda_h$, whose precise definition we do not need.
	\end{definition}
%%%%%%%%%%%%%%%%%%%%%%%%%%%%%%%%%%%%%%%%%%%%%%%%%%%%%%%%%%%%%%%%%%%%%%	

	\begin{lemma}
	\label{thm:project_Omega_to_Lambda_h}
	Let $\Omega$ satisfy \zcref{ass:Omega}. Then there exist positive constants $h_0 \equiv h_0(\Omega)$ and $c_0 \equiv c_0(\Omega)$ such that for all $h \in (0, h_0]$
	\begin{equation}\label{eq:pi_h_bound}
		\abs{\pi_h(x) - x} \leq c_0 h \quad \text{for all } x \in \Omega.
	\end{equation}
	\end{lemma}
	\begin{proof}
	Since $\Omega$ has a Lipschitz boundary, it satisfies the uniform interior cone condition: there exist $R \equiv R(\Omega) > 0$ and $\alpha \equiv \alpha(\Omega) \in (0,\pi/2) $ depending only on $\Omega$ and  a function $\nu \colon \Omega \to \S^{d-1}$ such that for all $x \in \Omega$
	\begin{equation}
		C_{R,\alpha}(x) := \set{x + t y \mid t \in (0,R), ~y \in \S^{d-1} \colon y \cdot \nu(x) > \cos\alpha} \subset \Omega.
	\end{equation}
	Set $r_0 := \frac{\sqrt d}{\sin \alpha}>0$ and $h_0 := \tfrac{R}{r_0 + \sqrt d}$. Then for all $0 < h \leq h_0$ one has that
	\begin{equation}\label{eq:ball_around_z_in_cone}
		B(x + r_0 h \nu(x); h \sqrt d) \subset C_{R,\alpha}(x).
	\end{equation}
	The radius of the ball is chosen such that $h\Z^d \cap B(x + r_0 h \nu(x); h \sqrt d) \neq \emptyset$. Thus, there exists $\hat x \in \Lambda_h$ such that $\abs{x - \hat x} \leq (r_0 + \sqrt d) h$. In particular, this implies \eqref{eq:pi_h_bound} with $c_{0}=r_0 + \sqrt d$.
	\end{proof}
	
	\begin{lemma}\label{thm:existence_discretization}
		Let $\Omega$ satisfy \zcref{ass:Omega}. Then there exist positive constants $h_0 \equiv h_0(\Omega)$ and $c_0 \equiv c_0(\Omega)$ such that for all $h \in (0, h_0]$
		and any $\gamma \colon [0,1] \to \Omega$ continuous with $\gamma(0), \gamma(1) \in \Lambda_h$, 
		there exists $\hat \gamma \colon [0,1] \to \Lambda_h$ satisfying the following properties:
		\begin{enumerate}[label=(\roman*)]
			\item $\gamma(0) = \hat\gamma(0)$ and $\gamma(1) = \hat\gamma(1)$, \label{item:existence_discretization_endpoints}	
			\item it is piecewise constant with only finitely many discontinuities at $0 < t_1 < \dots < t_k < 1$, \label{item:existence_discretization_pw_constant} 
			\item it satisfies $\norm{\gamma - \hat\gamma}_\infty := \sup_{s \in [0,1]} \abs{\gamma(s) - \hat\gamma(s)} \leq c_0 h$, and \label{item:existence_discretization_error}		
			\item it is right continuous and satisfies $\abs{\hat\gamma(t_{i}) - \hat\gamma(t_{i-1})} \leq c_0 h$ $\forall i \in \set{1, \dots, k}$, where $t_0 = 0$. \label{item:existence_discretization_continuity}
		\end{enumerate}
	\end{lemma}
	\begin{proof}
		Let $C([0,1]; \Omega)$ be the space of continuous functions $[0,1] \to\Omega$ with the supremum norm. Since the set of piecewise linear function $[0,1] \to\Omega$ is dense in $C([0,1]; \Omega)$, there exists a piecewise linear function $\tilde\gamma \colon [0,1] \to \Omega$ such that $\norm{\tilde\gamma - \gamma}_{\infty} \leq h$. Without loss of generality, we assume $\tilde\gamma(0)=\gamma(0)$ and $\tilde\gamma(1)=\gamma(1)$, because one way of constructing such a $\tilde\gamma$ is linearly interpolating sufficiently many points in $\gamma([0,1]) \subset \Omega$.
		
		Let $\pi_h \colon \Omega \to \Lambda_h$ be the map defined in \eqref{eq:def_pi_h} and let $\hat\gamma := \pi_h \circ \tilde\gamma$. We claim that $\hat\gamma$ satisfies the requirements of the statement after possibly modifying the definition for a finite number of points $t_i \in [0,1]$.
		
		Since $\tilde\gamma(0) = \gamma(0) \in \Lambda_h$ and $\tilde\gamma(1) = \gamma(1) \in \Lambda_h$, \zcref[noname]{item:existence_discretization_endpoints} holds by definition of $\pi_h$. 
		To show \zcref[noname]{item:existence_discretization_pw_constant}, note that for any $x \in \Lambda_h$ the set 
		\begin{equation}
			A_x := \set{y \in \R^d \mid \abs{y - x} \leq \abs{y - z} \text{ for all } z \in \Lambda_h}
		\end{equation}
		is an intersection of half planes and thus convex. 
		Therefore, every straight line segment of $\tilde\gamma$ (of which there are only finitely many) intersect $A_x$ at most twice. Since the only points where $\hat\gamma$ can be discontinuous are at these intersection points ($\hat\gamma(s)=\hat\gamma(s')$ for all $s,s' \in [0,1]$ such that $\tilde\gamma(s), \tilde\gamma(s') \in A_x$), it can have at most finitely many points of discontinuity.
		
		Let $0 \leq t_1 < \ldots < t_k \leq 1$ be the points of discontinuity of $\hat\gamma$. We claim that $t_1 > 0$. Indeed, since $\tilde\gamma(0) \in \Lambda_h$, it holds that $B(\tilde\gamma(0); h/2) \subset A_{\tilde\gamma(0)}$. By continuity of $\tilde\gamma$, there exists $\delta > 0$ such that 
		\begin{equation}
			\tilde\gamma([0,\delta)) \subset B(\tilde\gamma(0); h/2) \subset A_{\tilde\gamma(0)}.
		\end{equation}
		This implies for all $s \in [0,\delta)$
		\begin{equation}
			\hat\gamma(s) = \pi_h(\tilde \gamma(s)) = \tilde\gamma(0),
		\end{equation}
		and therefore $t_1 \geq \delta$. 
		Analogously, we must have $t_k < 1$.
		
		We redefine (without relabeling) $\hat\gamma$ such that $\lim_{s \downarrow t} \hat\gamma(s) = \hat\gamma(t)$ for all $t_1, \ldots, t_k \in (0,1)$ which makes $\hat\gamma$ right continuous and does not affect \zcref[noname]{item:existence_discretization_endpoints,item:existence_discretization_pw_constant}.
		
		We proceed with the proof of \zcref[noname]{item:existence_discretization_error}. By \zcref{thm:project_Omega_to_Lambda_h}, there exist $h_0(\Omega) >0$ and $c_0(\Omega) > 0$ such that for all $0 < h \leq h_0$
		\begin{equation}
			\abs{\pi_h(x) - x} \leq c_0 h \quad \text{for all } x \in \Omega.
		\end{equation}
		We assume for the rest of the proof that $h \leq h_0$. If $t \in [0,1] \setminus \set{t_1, \ldots, t_k}$, then
		\begin{equation}
			\abs{\hat\gamma(t) - \gamma(t)} 
			\leq \abs{\hat\gamma(t) - \tilde\gamma(t)} + \abs{\tilde\gamma(t) - \gamma(t)}
			= \abs{\pi_h(\tilde\gamma(t)) - \tilde\gamma(t)} + \abs{\tilde\gamma(t) - \gamma(t)}
			\leq (c_0 + 1) h.
		\end{equation}
		On the other hand, if $t \in \set{t_1, \ldots, t_k}$, then
		\begin{equation}
			\abs{\hat\gamma(t) - \tilde\gamma(t)}
			= \lim_{\eps \downarrow 0} \abs{\hat\gamma(t + \eps) - \tilde\gamma(t + \eps)}
			= \lim_{\eps \downarrow 0} \abs{\pi_h(\tilde\gamma(t + \eps)) - \tilde\gamma(t + \eps)}
			\leq c_0 h
		\end{equation}
		by the right continuity of $\hat\gamma$ and $\tilde\gamma$ and \zcref{thm:project_Omega_to_Lambda_h}. 
		
		Finally, for \zcref[noname]{item:existence_discretization_continuity}, note that for all $i \in \set{1, \dots, k}$
		\begin{equation}
			\abs{\hat\gamma(t_{i-1}) - \hat\gamma(t_i)}
			= \lim_{\eps \downarrow 0} \abs{\pi_h(\tilde\gamma(t - \eps)) - \pi_h(\tilde\gamma(t + \eps))}
			\leq 2 c_0 h + \lim_{\eps \downarrow 0} \abs{\tilde\gamma(t-\eps) - \tilde\gamma(t+\eps)}.
		\end{equation}
		The last term vanishes by continuity of $\tilde\gamma$.
	\end{proof}

	With all technical tools at hand, we proceed with proving the main theorem of this section.
	\begin{proof}[Proof of \zcref{thm:irr_random_graph}]
		If $\supp \eta = [0,\infty)$, both claims are trivial, because in that case $\G_{n,\eps}$ is the complete graph (with positive edge weights) which has $\irr(\G_{n, \eps}) = 3$ (cf. \zcref{rem:irrG}). Therefore, from now on we will assume that $\supp \eta = [0, 2\bar\eta]$ is compact.
		
		We define the event that every cube $\square_h(v)$ for $v \in \Lambda_h$ contains at least one element of $\V_n$
		\begin{equation}
			\mathcal A_h := \set{\omega \in \Xi \colon \abs{\V_n(\omega) \cap \square_h(v)} \geq 1 \text{ for all } v \in \Lambda_h}.
		\end{equation}
		By \zcref{thm:each_square_has_k_points} we have for all $h > 0$ with $nh^d \geq 1$
		\begin{equation}
			\P\left(\mathcal A_h\right)
			\geq 1 - C_1 n \expa{-C_2 n h^d},
		\end{equation}
for some constants $C_{1},C_{2}>0$. Let $h_0 > 0$ and $c_0 > 0$ be the constants from \zcref{thm:existence_discretization}
and set $h \equiv h(\eps) = c_1 \bar\eta \eps$, where we choose $c_1 := \frac{1}{2c_0 + \sqrt d+1}$. Moreover, we choose $\eps_0 \equiv \eps_0(\Omega) > 0$ such that $h(\eps) \leq h_0$ for all $0 < \eps \leq \eps_0$.
		The rest of the proof will be devoted to proving
		\begin{equation}\label{eq:Gconnirr}
			\set{{\G}_{n,\eps} \text{ is connected and } \irr\left({{\G}}_{n,\eps}\right)\leq K} \supset \mathcal A_h\qquad 
			\text{for all } h \equiv h(\eps)\leq h_0.
		\end{equation}
		Fix $\omega \in \mathcal A_h$. For any $v \in \Lambda_h$, choose arbitrarily a vertex
		$v' \in \V_{n}(\omega) \cap \square_h(v)$ and denote this assignment by $\Psi_h \equiv \Psi_h(\omega) \colon \Lambda_h \to \V_n(\omega)$. Since  $\omega \in \mathcal A_h$,
		this map is well defined and injective. Let $\gamma \colon [0,1] \to \Omega$ be a continuous curve such that $\gamma(0), \gamma(1) \in \Lambda_h$, which exists since $\Omega $ is connected.  Let $\hat\gamma \colon [0,1] \to \Lambda_h$ be the discretization of $\gamma$ from \zcref{thm:existence_discretization} with points of discontinuity $0 < t_1 <\ldots < t_k < 1$.
		We define the sequence of points in $\V_n(\omega )$ (cf. \zcref{fig:Pgamma})
		\begin{equation}\label{eq:def_Pgamma}
			P_\gamma := \Psi_h(\hat\gamma(0)) \, \Psi_h(\hat\gamma(t_1)) \cdots \Psi_h(\hat\gamma(t_k)).
		\end{equation}
		We show  that $P_\gamma$ is  a path in $\G_{n, \eps}(\omega)$, i.e. $P_\gamma \in \mathbf E(\G_{n, \eps}(\omega))$.
		If $\hat{\gamma}$ is constant, there are no discontinuity points and $P_{\gamma}$ reduces to a single point which we identify with the $0$ element in $\mathbf E(\G_{n, \eps}(\omega))$. If $k\geq 1$ we need to prove  $\set{\Psi_h(\hat\gamma(t_{i})), \Psi_h(\hat\gamma(t_{i-1}))} \in \E_{n,\eps}(\omega)$ for all $i \in \set{1, \dots, k}$  (with the convention $t_0 := 0$). In this regard, note that 
		\begin{equation}
			\abs{\Psi_h(\hat\gamma(t_{i})) - \Psi_h(\hat\gamma(t_{i-1}))} 
			\leq (\sqrt d+c_0) h \leq c_1^{-1} h = {\bar\eta \eps},
		\end{equation}
		where we used that $\abs{\Psi_h(z)-z} \leq \frac{h \sqrt d}2$ for all $z \in \Lambda_h$ and \zcref{thm:existence_discretization}\zcref[noname]{item:existence_discretization_continuity}. Therefore,
		\begin{equation}	
			\eta_\eps(\abs{\Psi_h(\hat\gamma(t_{i})) - \Psi_h(\hat\gamma(t_{i-1}))}) 			
			= \eps^{-d} \eta\left(\frac{\abs{\Psi_h(\hat\gamma(t_{i})) - \Psi_h(\hat\gamma(t_{i-1}))}}{\eps}\right)
			\geq \eps^{-d} \eta\left({\bar\eta}\right)
			> 0,
		\end{equation}
		so that these two points are indeed connected by an edge.
		
		\begin{figure}
			\centering
			\def\svgwidth{0.7\textwidth}
			\begingroup
			\makeatletter
			\newcommand*\fsize{\dimexpr\f@size pt\relax}
			\newcommand*\lineheight[1]{\fontsize{\fsize}{#1\fsize}\selectfont}
			\setlength{\unitlength}{\svgwidth}
			\global\let\svgwidth\undefined
			\makeatother
			\begin{picture}(1,0.40854425)
				\lineheight{1}
				\setlength\tabcolsep{0pt}
				\put(0,0){\includegraphics[width=\unitlength,page=1]{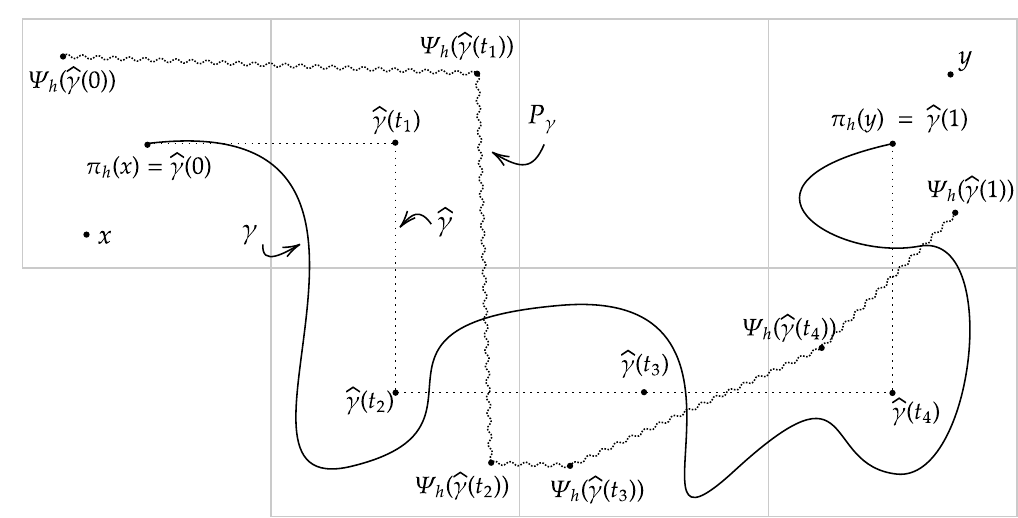}}
			\end{picture}
			\endgroup
			%\includesvg[width=0.7\textwidth]{img/Pgamma}			
			\caption{Construction of the path $P_\gamma$: Start with a continuous curve $\gamma$ (solid line) and discretize it to obtain $\hat\gamma$ (dotted). Finally  $P_\gamma$ is obained projecting $\hat\gamma$ back to $\G_{n, \eps}$ via the map $\Psi_h$ (dotted wavy line).}
			\label{fig:Pgamma}
		\end{figure}
		
		\medskip
		We are now ready to prove \eqref{eq:Gconnirr} and start by showing that $\G_{n, \eps}(\omega)$ is connected.
		Let $x,y \in \V_n(\omega)$ and let $\hat x := \pi_h(x)$ and $\hat y := \pi_h(y)$. There exists $\gamma \colon [0,1] \to \Omega$ continuous, such that $\gamma(0) = \hat x$ and $\gamma(1) = \hat y$ and $\gamma$ induces the path in $\G_{n,\eps}(\omega)$ given by $	P_{\gamma} = \Psi_h(\hat x) \cdots \Psi_h(\hat y) \in \mathbf E(\G_{n, \eps}(\omega))$.
		Further, $\abs{x - \Psi_h(\hat x)} \leq (\sqrt d + c_0) h \leq \bar\eta \eps$, thus if $x \neq \Psi_h(\hat x)$ then $\set{x, \Psi_h(\hat x)} \in \E_{n, \eps}(\omega)$. Similarly, either $y = \Psi_h(\hat y)$ or $\set{y, \Psi_h(\hat y)} \in \E_{n, \eps}(\omega)$.
		Therefore, $x P_\gamma y \in \mathbf E(\G_{n, \eps}(\omega))$, which shows that $x$ and $y$ are connected in $\G_{n, \eps}(\omega)$.
		
		\medskip To prove \eqref{eq:Gconnirr}  it remains to show  $\irr(\G_{n,\eps}(\omega)) \leq K$ for a constant $K$ which depends only on $\Omega$ and $\bar\eta$.
		Let $f \in \mathbf C(\G_{n,\eps}(\omega))$. It suffices to show that there exist $f_1, \ldots f_m \in \mathbf C(\G_{n,\eps}(\omega))$ such that $f = f_1 + \ldots + f_m$ and $\abs{f_i} \leq K$ for every $i \in \set{1, \ldots, m}$. Without loss of generality we assume that $f$ represents a closed cycle in $\G_{n,\eps}(\omega)$, otherwise we decompose every connected component separately  (cf.  \zcref{rem:cyclespace}). Therefore, we view $f$ as a path in $\G_{n, \eps}(\omega)$ given by
		\begin{equation}
			f = u_0 \cdots u_m u_0, \quad u_i \in \V_{n}(\omega).
		\end{equation}
		Assume for the moment the following: there exists $\gamma \colon [0,1] \to \Omega$ continuous with $\gamma(0) = \gamma(1) \in \Lambda_h$ which satisfies
		\begin{equation}\label{eq:additional_assumption}
			f = P_\gamma.
		\end{equation}
		Since $\Omega$ is simply connected, there exists a continuous function $H \colon [0,1]^2 \to \Omega$ such that $H(0, \cdot) = \gamma$, $H(1, \cdot) = \gamma(0)$ and $H(\cdot, 0) = H(\cdot, 1) = \gamma(0)$ contracting $\gamma$ to the point $\gamma(0)$. By uniform continuity of $H$, there exists $M$ large enough, such that
		\begin{equation}
			\norm{H\left(\frac{k}{M}, \cdot\right) - H\left(\frac{k-1}{M}, \cdot\right)}_\infty < h
		\end{equation}
		for all $k \in \set{1, \ldots, M}$. We define $\gamma^k := H\left(\frac{k}{M}, \cdot\right)$ with its discretization $\hat\gamma^k \colon [0,1] \to \Lambda_h$ as in \zcref{thm:existence_discretization}. Each $\gamma^k$ defines a cycle $f^k := P_{\gamma^k} \in \mathbf{C}(\G_{n, \eps}(\omega))$ by the construction \eqref{eq:def_Pgamma}.
		Note that it is possible that $\hat\gamma^k$ is constant, which implies  $f^k = 0$. In particular $f^{M}=0$.
		We have
		\begin{equation}
			f = f^0 + 2 \sum_{k=1}^M f^k = (f^0+f^1) + \sum_{k=2}^M (f^k + f^{k-1}).
		\end{equation}
		where we used that $f = f^0$ and that addition in $\mathbf E(\G_{n, \eps}(\omega))$ is defined $\textrm{mod}~2$.
		We show
		\begin{equation}\label{eq:f+f_0_decomposition}
			f^0 + f^1 = \sum_{l=1}^L g_l,
		\end{equation}
		for some $g_l \in \mathbf C(\G_{n, \eps}(\omega))$ with $\abs{g_l} \leq 4$.  By the same arguments,  each sum $f^k + f^{k-1}$ can be decomposed into a sum of cycle of length at most four, which concludes the proof under the additional assumption \eqref{eq:additional_assumption}. 
		Let now
		\begin{equation}
			\set{\tau_1, \dots, \tau_L}
			:= \set{s \in (0,1) \mid \hat\gamma^0 \text{ or } \hat\gamma^1 \text{ are discontinuous at $s$}}
		\end{equation}
		be the joint points of discontinuity of $\hat\gamma^0$ and $\hat\gamma^1$, ordered such that $0 < \tau_1 < \dots < \tau_L < 1$ and set $\tau_0 : =0$. 
		For $l \in \set{1,\dots, L}$ define the path
		\begin{equation}
			g_l := \Psi_h(\hat\gamma^0(\tau_{l-1})) \,\Psi_h(\hat\gamma^0(\tau_l)) \, \Psi_h(\hat\gamma^1(\tau_l)) \, \Psi_h(\hat\gamma^1(\tau_{l-1})) \, \Psi_h(\hat\gamma^0(\tau_{l-1})),
		\end{equation}
		where we ignore repeated, consecutive vertices. 
		Clearly, $\abs{g_l} \leq 4$, so it remains to show $\gamma_l \in \mathbf C(\G_{n, \eps}(\omega))$.
		If $\hat\gamma^0(\tau_{l-1}) \neq \hat\gamma^0(\tau_l)$, then $\set{\Psi_h(\hat\gamma^0(\tau_{l-1})), \Psi_h(\hat\gamma^0(\tau_l))} \in \E_{n, \eps}(\omega)$ as in the construction of $P_\gamma$ \eqref{eq:def_Pgamma}. Similarly, if $\hat\gamma^1(\tau_{l-1}) \neq \hat\gamma^1(\tau_l)$, then $\set{\Psi_h(\hat\gamma^1(\tau_{l-1})), \Psi_h(\hat\gamma^1(\tau_l))} \in \E_{n, \eps}(\omega)$.	
		Hence, it remains to show
		\begin{equation}
			\set{\Psi_h(\hat\gamma^0(\tau_l)), \Psi_h(\hat\gamma^1(\tau_l))} \in \E_{n, \eps}(\omega)
			\quad\text{and}\quad
			\set{\Psi_h(\hat\gamma^0(\tau_{l-1})), \Psi_h(\hat\gamma^1(\tau_{l-1}))} \in \E_{n, \eps}(\omega).
		\end{equation}
		This holds, because
		\begin{equation}
			\begin{aligned}
				\abs{\Psi_h(\hat\gamma^0(\tau_{l})) - \Psi_h(\hat\gamma^1(\tau_{l}))}
				&\leq h \sqrt d + \norm{\hat\gamma^0 - \gamma^0}_\infty + \norm{\hat\gamma^1 - \gamma^1}_\infty +\abs{\gamma^0(t_{l})- \gamma^1(t_{l})}\\
				&\leq(\sqrt d + 2 c_0 + 1)h = c_1^{-1} h = {\bar\eta \eps}
			\end{aligned}
		\end{equation}
		and analogously for the second edge.
		Finally, one readily verifies \eqref{eq:f+f_0_decomposition} (cf. \zcref{fig:build_squares}).
		
		\begin{figure}
			\centering			
			\def\svgwidth{0.5\textwidth}
			\begingroup
			\makeatletter
			\newcommand*\fsize{\dimexpr\f@size pt\relax}
			\newcommand*\lineheight[1]{\fontsize{\fsize}{#1\fsize}\selectfont}
			\setlength{\unitlength}{\svgwidth}
			\global\let\svgwidth\undefined
			\makeatother
			\begin{picture}(1,0.40854425)
				\lineheight{1}
				\setlength\tabcolsep{0pt}
				\put(0,0){\includegraphics[width=\unitlength,page=1]{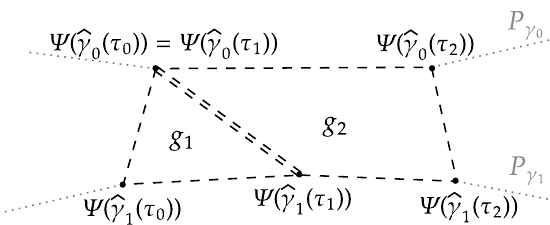}}
			\end{picture}
			\endgroup
			%\includesvg[width=0.5\textwidth]{img/build_squares}			
			\caption{Decomposing $P_{\gamma_0} + P_{\gamma_1}$ for two curves $\gamma_0, \gamma_1$ which are close in $C([0,1], \Omega)$. The edge $\set{\Psi(\hat\gamma_0(\tau_0)), \Psi(\hat\gamma_1(\tau_1))}$ appears in $g_1$ and $g_2$, so it does not count in the decomposition.}
			\label{fig:build_squares}
		\end{figure}
		
		\medskip It remains to remove the additional assumption \eqref{eq:additional_assumption}.
		Recall that
		\begin{equation}\label{eq:f_as_path_of_vertices}
			f = u_0 \cdots u_m u_0
		\end{equation}
		and that we assumed that $\eta$ has compact support $\supp \eta = [0, 2\bar\eta]$. In particular for any edge $\{u,v \}\in \E_{n,\varepsilon} (\omega)$ we have $|u-v|\leq 2\bar \eta$.
Therefore,
		\begin{equation}
			\abs{\pi_h(u_i) - \pi_h(u_{i-1})} 
			\leq \abs{\pi_h(u_i) - u_{i}} + \abs{\pi_h(u_{i-1}) - u_{i}}  + \abs{u_i - u_{i-1}} 
			\leq  2c_{0}h+ 2\bar\eta \eps\leq 3 \bar \eta \eps
		\end{equation}
		for all $i \in \set{1,\dots,m+1}$ with the convention $u_{m+1}=u_{0}$.
		Since $\Omega$ has a Lipschitz boundary, there exist constants $L \equiv L(\Omega) > 0$ and $0<\eps_1(\Omega) \leq \varepsilon_{0}$, such that we can connect $\pi_h(u_{i-1})$ and $\pi_h(u_{i})$ by a curve $\gamma_i \colon [0,1] \to \Omega$ with $\gamma_i(0)=\pi_h(u_{i-1})$ and $\gamma_i(1)=\pi_h(u_i)$ with
\[
\ell(\gamma_i) \leq L|\gamma_{i} (0)-\gamma_{i} (1) |\leq 3 L \bar\eta \eps \quad \text{for all } \eps \leq \eps_1(\Omega),
\]
where $\ell(\gamma_i)$ denotes the length of $\gamma_i$. Let $\gamma \colon [0,1] \to \Omega$ be the concatenation of the curves $\gamma_1, \dots, \gamma_m$ such that $\gamma(i/(m+1)) = \pi_h(u_i)$ for $i \in \set{0,\dots,m}$ and $\gamma(1)=\gamma(0) = \pi_h(u_0)$ and define $g := P_\gamma$. Then
		\begin{equation}
			f = f +g + g
		\end{equation}
		We claim that there exists $K \equiv K(\Omega, \bar\eta) \geq 4$ such that
		\begin{equation}
			f + g = \sum_{j=1}^m h_j
		\end{equation}
		for some $h_j \in \mathbf C(\G_{n, \eps}(\omega))$ with $\abs{h_j} \leq K$. This implies  $\irr(\G_{n, \eps}(\omega)) \leq K$, because $g$ can be decomposed as the sum of cycles of length at most four by the arguments above.
To prove the claim, let $v_i := \Psi_h(\pi_h(u_i))$ for $i \in \set{0, \dots, m}$. Then we can write
		\begin{equation}
			g = P_0 + P_1 + \dots + P_m 
		\end{equation}
		where $P_i \in \mathbf E(\G_{n, \eps}(\omega))$ is the part of the path $P_{\gamma}$ which connects $v_i$ and $v_{i+1}$ (cf. \zcref{fig:boundary_paths}). Then
		\begin{equation}
			f + g = \sum_{j=0}^m u_j P_j u_{j+1} u_j,
		\end{equation}
		where we used that 
		\begin{equation}
			\abs{u_j - v_j} 
			= \abs{u_j - \Psi_h(\pi_h(u_j))}
			\leq \left(\frac{\sqrt d}2 + c_0\right) h \leq \bar\eta \eps
		\end{equation}
and hence either $u_{j}=v_{j}$ or $\{u_{j},v_{j} \}$ is an edge and in this case
it is summed over twice. 
		Hence, $h_j := u_j P_j u_{j+1} u_j \in \mathbf C(\G_{n, \eps}(\omega))$. It remains to prove $|h_{j}|\leq K$
		for some $K=K (\Omega,\bar \eta)$. For this we show
\[
P_{j}\subset B (v_{j}, Rh)
\]
for some $R=R (\Omega,\bar \eta)>0$.  
Indeed, we can write any vertex $z \in P_j$ as $z=\Psi_h(\hat{\gamma}_{j} (t))$ for some $t \in [0,1]$. Then
\begin{align*}
	\abs{z - v_j} &\leq 
			 h \sqrt d + \abs{\hat{\gamma}_{j} (t)) - \pi_h(u_j)}\leq 
			 h \sqrt d+c_{0}h+ \abs{\gamma_{j} (t)) - \pi_h(u_j)}\\
&=  h \sqrt d+c_{0}h+ \abs{\gamma_{j} (t)) - \gamma_{j} (0)}\leq 
			h \sqrt d + \ell(\gamma_j) + c_0 h
			\leq R h
		\end{align*}
		for a constant $R \equiv R(\Omega, \bar\eta) > 0$.
		A ball of radius $R h$ intersects at most $(2 R + 1)^d$ disjoint cubes of side length $h$,
		so $P_j$ consists of at most $(2 R + 1)^d$ vertices. There could potentially be an edge between all of them, and additionally the three edges containing $u_j$ and $u_{j+1}$, so that
		\begin{equation}
			\abs{h_j} \leq (2 R + 1)^{2d} + 3 =: K.
		\end{equation}
		As this constant depends only on $\Omega$ and $\bar\eta$, this concludes the proof that $\irr(\G_{n, \eps}(\omega)) \leq K$.\qedhere
		\begin{figure}
			\centering			
			\def\svgwidth{0.6\textwidth}
			\begingroup
			\makeatletter
			\newcommand*\fsize{\dimexpr\f@size pt\relax}
			\newcommand*\lineheight[1]{\fontsize{\fsize}{#1\fsize}\selectfont}
			\setlength{\unitlength}{\svgwidth}
			\global\let\svgwidth\undefined
			\makeatother
			\begin{picture}(1,0.40854425)
				\lineheight{1}
				\setlength\tabcolsep{0pt}
				\put(0,0){\includegraphics[width=\unitlength,page=1]{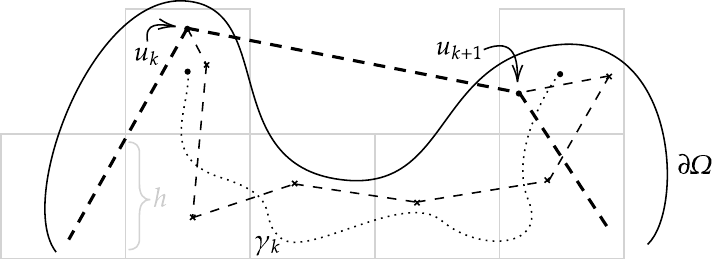}}
			\end{picture}
			\endgroup
			%\includesvg[width=0.6\textwidth]{img/boundary_paths}
			\caption{Decomposition of $f \in \mathbf C(\G_{n,\eps}(\omega))$ close to the boundary $\partial \Omega$ as in \eqref{eq:f_as_path_of_vertices}. The fat dashed line represents edges in $f$. The two points $u_k$ and $u_{k+1}$ are connected. The dotted curve $\gamma_k$ connects the points $\pi_h(u_k)$ and $\pi_h(u_{k+1})$ in $\Omega$. From $\gamma_k$, we construct a new path $P_k$ from $u_k$ to $u_{k+1}$ in $\G_{n,\eps}(\omega)$, which is represented by the thin dashed line, and which uses at most $K$ vertices from $\V_n(\omega)$. This path together with the edge $\set{u_k, u_{k+1}}$ is denoted by $h_k \in \mathbf C(\G_{n,\eps}(\omega))$ in the proof.}
			\label{fig:boundary_paths}
		\end{figure}
	\end{proof}

	\subsection{Proof of the main theorem}
	\label{sec:random_graphs}
	
	To prove our main theorem, we will combine \zcref{thm:irr_random_graph} with \zcref{thm:expectation_bound_fixed_graph}. This will require estimating the effective resistance (cf. \zcref{def:effres}) of the random graph. In this chapter all constants are assumed to depend only on $\Omega$, $\eta$ and $\rho$, even when we do not mention this explicitly.
	
	\begin{lemma}\label{thm:random_graph_effective_resistance}
		Assume that $\Omega$, $\eta$ and $\rho$ satisfy \zcref{ass:Omega,ass:eta, ass:rho} and that the coupling constant $\J_{n,\eps}$ are given by \eqref{eq:random_coupling_constants_normalized}. Additionally, assume $\supp \eta \subset [0,1]$ and $\norm{\eta}_{L^1([0,1])} = 1$.
				
		Then there exist positive constants $R \equiv R(\Omega, \eta, \rho)$, $c_1 \equiv c_1(\Omega)$, $c_2 \equiv c_2(\Omega, \eta(0), \rho_{\min}, \rho_{\max})$ and $\eps_0 \equiv \eps_0(\Omega, \rho, \eta)$ such that for any $n \geq 2$, $n^{-\frac1d} < \varepsilon \leq \eps_0$
		\begin{equation}\label{eq:bound_probabilitic_resistance}
			\P\left(\mathbf{R}_{\G_{n, \eps}} \leq R\eps^{-2}\right)
			\geq 1 - c_1 n e^{-c_2 n \eps^d},
		\end{equation}
		where $\mathbf{R}_{\G_{n, \eps}}$ is the effective resistance of $\G_{n, \eps}$ as defined in \eqref{eq:def_effective_resistance}.
	\end{lemma}
	\begin{proof}
		Define two events $\mathcal A_1, \mathcal A_2 \subset \Xi$ by
		\begin{equation}
			\begin{gathered}
				\mathcal A_1 := \set{\G_{n, \eps} \text{ is connected and for all } x \in \V_n \colon {\deg_{n,\eps}(x)} \leq n c_0},\\
				\mathcal A_2 := \set{\forall v \in \ell^2(\G_{n, \eps}) \colon \norm{v - (v)_{\deg_{n,\eps}}}_2
					\leq \frac{c_1}{\sigma_{\eta} (n-1) \varepsilon^2} 
					\sum_{\set{x,y} \in \E_{n,\eps}} \eta_\eps(\abs{x-y}) \left({v(x) - v(y)}\right)^2}.
			\end{gathered}
		\end{equation}
		where $c_0$ and $c_1$ are the constant from \eqref{eq:degree_bound} and \eqref{eq:discrete_poincare}. By \zcref{thm:degrees_bounded,thm:discrete_poincare,thm:irr_random_graph} it holds 
		\begin{equation}
			\P(\mathcal A_1 \cap \mathcal A_2)
			\geq 1-  c n e^{-c' n \eps^d}
		\end{equation}
		for some positive constants $c, c' > 0$.
		We show that there exists a constant $R \equiv R(\Omega, \eta, \rho)$ such that \eqref{eq:bound_probabilitic_resistance} holds for all $\omega \in \mathcal A_1 \cap \mathcal A_2$. To estimate the effective resistance between $x$ and $y$ in $\G_{n, \eps}(\omega)$, let $u \in \ell^2({\G}_{n,\eps}(\omega))$ be the unique solution of
		\begin{equation}\label{eq:u_sol_f_knv}
			-\Delta_{\G_{n,\eps}(\omega)} u = e_x - e_y 
			\quad \text{and} \quad 
			(u)_{\deg_{n,\eps}^\omega} = 0.
		\end{equation}
		Then
		\begin{equation}\label{eq:bound_effective_resistance_tilde_G}
			\left(e_x - e_y, \left(-\Delta_{\G_{n,\eps}(\omega)}\right)^{-1} \left(e_x - e_y \right)\right)
			\leq \norm{e_x - e_y}_2 \norm{u}_2.
		\end{equation}
		Recall definition \eqref{eq:random_coupling_constants_normalized}
		\begin{equation}
			\J_{n,\eps}({x,y}) := \frac{\eta_\eps(\abs{x-y})}{\sqrt{\deg_{n,\eps}(x) \deg_{n,\eps}(y)}} \qquad \text{for } x,y \in \V_n.
		\end{equation}
		Since $u$ solves \eqref{eq:u_sol_f_knv}, using that $\omega \in \mathcal A_1 \cap \mathcal A_2$,
		\begin{equation}
			\begin{aligned}
				\norm{u}_2
				&\leq \frac{c_1}{\sigma_{\eta} (n-1) \varepsilon^2} \sum_{\set{x,y} \in \E_{n,\eps}(\omega)} \eta_\eps(\abs{x-y}) \left({u(x) - u(y)}\right)^2 \\
					&\leq \frac{c_1 \max_{x \in \V_n(\omega)} \deg_{n,\eps}^\omega(x)}{\sigma_{\eta} (n-1) \varepsilon^2}  \left(u, -\Delta_{\G_{n,\eps}(\omega)} u\right) \\
				&\leq \frac{c_2}{\sigma_{\eta}\varepsilon^2} \left(u, e_x - e_y\right)
				\leq \frac{c_2}{\sigma_{\eta} \varepsilon^2} \norm{e_x - e_y}_2 \norm{u}_2,
			\end{aligned}
		\end{equation}
		where $c_2 = 2 c_0 c_1$.
		Dividing by $ \norm{u}_2$ on both sides and inserting into \eqref{eq:bound_effective_resistance_tilde_G}, we obtain
		\begin{equation}\label{eq:final_bound_two_point_effres}
			\left(e_x - e_y, \left(-\Delta_{\G_{n,\eps}(\omega)}\right)^{-1} \left(e_x - e_y \right)\right)
			\leq \frac{c_2}{\sigma_{\eta} \eps^2} \norm{e_x - e_y}_2^2 
			= \frac{2 c_2}{\sigma_{\eta} \eps^2},
		\end{equation}
		where we used that $\norm{e_x - e_y}_2^2 = 2$ for the last bound. Since	the discrete Poincaré inequality holds for all $v \in \ell^2(\G_{n, \eps})$ (on the event $\mathcal A_2$), we can take the supremum over $x,y \in \V_n$ in \eqref{eq:final_bound_two_point_effres}, which concludes the proof of \eqref{eq:bound_probabilitic_resistance} with $R := \frac{2 c_2}{\sigma_\eta}$.
	\end{proof}

%%%%%%%%%%%%%%%%%%%%%%%%%%%%%%%%%%%%%%%%%%%%%%%%%%%%%%%%%%%%%%%%%%%%%%%%%%%%%%%%%%%%%%%%%%%%%%%%%%%%%%%%%%
	With these probabilistic estimates on the effective resistance of $\G_{n, \eps}$ and the non-probabilistic bound from \zcref{thm:expectation_bound_fixed_graph}, we prove now the quenched lower bound on the two-point function.

	\begin{proof}[Proof of \zcref{thm:main}] 
		We assume first that $\eta={\bar \eta^{-1}} \I_{[0,\bar{\eta}]}$ holds for some $\bar \eta\in (0,1]$. Consider the two events $\mathcal{A}_{1},\mathcal{A}_{2}\subset \Xi$
		given by
		\begin{equation}\label{eq:A1A2}
			\begin{gathered}
				\mathcal{A}_{1}=\set{\G_{n, \eps} \text{ is connected and } \irr\left({{\G}}_{n,\eps}\right)\leq K \text{ and for all } x \in \V_n \colon \deg_{n,\eps}(x) \leq n c_0 }, \\
				\mathcal{A}_{2}=\set{\mathbf{R}_{\G_{n, \eps}} \leq R\eps^{-2}},
			\end{gathered}
		\end{equation}
		where $c_0$, $K$ and $R$ are the constants from \eqref{eq:degree_bound}, \eqref{eq:prob_connected_and_irr_bounded} and \eqref{eq:bound_probabilitic_resistance}. By \zcref{thm:degrees_bounded,thm:irr_random_graph,thm:random_graph_effective_resistance} we have
		\begin{equation}
			\P (\mathcal{A}_{1}\cap \mathcal{A}_{2})\geq 1-  c n e^{-c' n \eps^d}
		\end{equation}
		for some constants $c, c'>0$.
		Let  $\omega \in\mathcal{A}_{1}\cap \mathcal{A}_{2}$ and $\kappa \in (0, \tfrac{\pi}{K})$ a parameter. By \zcref{thm:expectation_bound_fixed_graph}, it holds
		\begin{equation}\label{eq:lowerboundindicator}	
					\sprod{\cos(\theta_x - \theta_y)}{\G_{n,\varepsilon} (\omega);\beta} 
					\geq 1
					- \tfrac{R}{2  \eps^2 \beta \cos \kappa}
					- c_1 n^{2} \, e^{-\frac{\beta J_{0}  \kappa^2}{\pi}},
		\end{equation}
		where we bounded the number of edges $\abs{\E_{n,\eps}(\omega)}\leq n^{2}$ and $J_0 \equiv J_0(\omega)$ is defined in \eqref{def:J0}. Since $\omega \in \mathcal A_1$, we have
		\begin{equation}
			J_{0}
			= \min_{\{i,j \}  \in \E_{n, \eps}(\omega)} \J_{n,\eps}^\omega(i,j) 
			\geq \frac{1}{\bar \eta \varepsilon^{d}} \min_{x  \in \V_{n}(\omega)} \deg^\omega_{n,\eps}(x)^{-1}
			\geq \frac{c_2}{n \eps^d},
		\end{equation}
		where $c_2 = \frac{1}{c_0 \bar \eta}$. This concludes the proof for the special choice of $\eta$. It remains to show how we can reduce a general $\eta$ to this special case.
		
		By continuity of $\eta$ at the origin (cf. \zcref{ass:eta}), we can choose $\bar\eta \in (0,1]$ such that $\eta(\bar\eta) \geq \eta(0)/2$. Define $\hat\eta := \bar\eta^{-1} \I_{[0, \bar\eta]}$. Let $\hat\G_{n,\eps} = (\V_n, \hat\E_{n,\eps}, \hat\J_{n,\eps})$ be defined in the same way as $\G_{n,\eps}$, with the role of $\eta$ replaced by $\hat\eta$ (see \zcref{sec:main_results}). We denote by $\widehat{\deg}_{n,\eps}$ the corresponding degree function \eqref{eq:def_degree}. This new graph has less edges, $\hat\E_{n,\eps} \subseteq \E_{n,\eps}$, but the edge weights may be larger.
		
		Since $\J_{n,\eps}$ is invariant under multiplication of $\eta$ by a constant, we rescale it (without relabeling and without modifying $\bar\eta$) such that $\eta \geq \hat\eta$. This does not change the coupling constants $\J_{n,\eps}$, so it does not affect the graph $\G_{n,\eps}$, but it guarantees $\deg_{n,\eps} \geq  \widehat{\deg}_{n,\eps}$.
		
		We define, with the constants $c_3$ and $c_4$ given in \eqref{eq:degree_bound}, the event $\mathcal A_3 \subset \Xi$
		\begin{equation}
			\mathcal A_3 
			:= \set{\text{for all } x \in \V_n \colon c_3 n \leq \widehat{\deg}_{n,\eps}(x) \leq \deg_{n,\eps}(x) \leq c_4 n}.
		\end{equation}
		On this event, it holds for all $x,y \in \V_n$
		\begin{equation}
			\J_{n,\eps}(x,y) 
			= \frac{\eta_\eps(\abs{x-y})}{\sqrt{\deg_{n,\eps}(x)\deg_{n,\eps}(y)}}
			\geq \sqrt{\frac{{\widehat{\deg}_{n,\eps}(x) \widehat{\deg}_{n,\eps}(y)}}{\deg_{n,\eps}(x)\deg_{n,\eps}(y)}} \hat\J_{n,\eps}(x,y)
			\geq \frac{c_4}{c_3} \hat\J_{n,\eps}(x,y).
		\end{equation}
		Therefore, the monotonicity of the two point function in the coupling constants implies
		\begin{equation}
		\sprod{\cos(\theta_x - \theta_y)}{\G_{n,\varepsilon} (\omega);\beta} 
		\geq \sprod{\cos(\theta_x - \theta_y)}{\hat\G_{n,\varepsilon} (\omega);\frac{c4}{c3}\beta}
		\end{equation}
		for all $\omega \in \mathcal A_3$.
		Finally, let $\hat{\mathcal{A}}_{1},\hat{\mathcal{A}}_{2}\subset \Xi$ be the events obtained from \eqref{eq:A1A2}
		by replacing $\G_{n,\varepsilon}$ with  $\hat{\G}_{n,\varepsilon}$  Then,
		\begin{equation}
			\P(\mathcal A_3 \cap \hat{\mathcal{A}}_{1}\cap \hat{\mathcal{A}}_{2}) \geq 1- c n e^{-c' n \eps^d}.
		\end{equation}
		By \eqref{eq:lowerboundindicator} we have, for all $\omega \in A_3 \cap \hat{\mathcal{A}}_{1}\cap \hat{\mathcal{A}}_{2}, $	
\[ \sprod{\cos(\theta_x - \theta_y)}{\hat\G_{n,\varepsilon} (\omega);\frac{c4}{c3}\beta}\geq 1 - \tfrac{\hat{R}}{  \eps^2 \beta}
					- c_1 n^{2} \, e^{-\frac{\beta \hat{c} }{n\varepsilon^{d}}},
\]
		from which the claim follows.
	\end{proof}
	
	\section{Alternative strategy: order in the interior for  random geometric graphs}
	\label{sec:interior}
	
	To conclude, we give an alternative version of \zcref{thm:main}, which holds in the interior of the domain $\Omega$. The idea of the proof is simple and relies only on a comparison of the \emph{$XY$}-model on the random graph $\G_{n, \eps}$ with the \emph{$XY$}-model on a deterministic subset of the lattice $\Z^d$. The two dimensional case will be treated slightly differently, because there is no long-range order in the $XY$-model on $\Z^2$.
	
	For $\delta > 0$ define the \emph{inner parallel set}
	\begin{equation}
		\Omega_{\delta} := \set{x \in \Omega \mid \dist(x, \partial \Omega) > \delta}.
	\end{equation}
	
	\begin{theorem}\label{thm:main_interior}
		
		Assume that $\Omega$, $\eta$ and $\rho$ satisfy \zcref{ass:Omega,ass:eta,ass:rho} and let $\delta > 0$.
		There exist positive constants $c_*$, $c_1$, $c_2$, $c_3$ and $\eps_0$, depending only on $\Omega$, $\rho$ and $\eta$, and a positive constant $a$ depending additionally on $\delta$,
		such that for all $\eps \in (0, \eps_0)$, $\beta \geq c_* n \eps^d,$  and $n \eps^{d} \geq c_1$ for $d\geq 3$ or  $n \eps^{3} \geq c_1$
		in the case $d=2,$ 
		the event that
		\begin{equation}\label{eq:long_range_order_graph_interior_paralell_set}
			\text{for all } x,y \in \V_n \cap \Omega_{\delta} \colon 
			\sprod{\cos{\theta_{x} - \theta_{y}}}{\G_{n,\eps};\beta} 
			\geq 1 - a\,\Gamma\left(\frac{\beta}{n\eps^d}\right)
		\end{equation}
		has probability at least $1 - c_2 n e^{-c_3 n \eps^d}$.
		Here $\Gamma(t) := \sqrt{\frac{\log t}{t}}$.
	\end{theorem}
	
	We will use the following pointwise lower bound on the two point function in finite subsets of $\Z^d$ which was established by Garban and Spencer in \cite{Garban:ContinuousSymmetryBreaking2022a}.
	\begin{theorem}[{\cite[Theorem 1.3 and Remark 1]{Garban:ContinuousSymmetryBreaking2022a}}]\label{thm:nishimori_long_range_order}
		Let $d \geq 3$ and recall $\Gamma(t) := \sqrt{\frac{\log t}{t}}$. There exist constants $\beta_*\equiv \beta_*(d) > 0$ and $a_* \equiv a_*(d) > 0$ such that for all $\beta \geq \beta_*$, all finite $\Lambda \subset \Z^d$, and all points $x,y \in \Lambda$ such that $B\left(\frac{x+y}2; 2\abs{x-y}_2\right)\cap \Z^{d} \subset \Lambda$
		\begin{equation}
			\sprod{\cos \theta_x - \theta_ y}{\Lambda; \beta} \geq 1 - a_*\, \Gamma (\beta ).
		\end{equation}
		Here, $\Lambda$ is seen as a finite subgraph of $\Z^d$ with nearest neighbor edge set, i.e. $\Lambda = (\Lambda, E_\Lambda, J_\Lambda)$, where
		\begin{equation}
			E_\Lambda := \big \{\set{i,j} \subset \Lambda \mid \abs{i-j}_1 = 1\big \}
			\quad \text{and} \quad
			J_\Lambda(i,j) := \begin{cases}
				1	& \text{if } \set{i,j} \in E_\Lambda,\\
				0	& \text{else}.
			\end{cases}
		\end{equation} 
	\end{theorem}
	
	\subsection{A comparison result}
	
	The main tool in establishing \zcref{thm:main_interior} will be the following comparison result, which will allow an application of the lower bound from \zcref{thm:nishimori_long_range_order}.
	
	\begin{proposition}\label{thm:comparison}
		Let $G_1=(V_1, E_1, J_1)$ and $G_2 = (V_2, E_2, J_2)$ be two weighted graphs, $\beta_1, \beta_2 > 0$. Let $\tau \colon V_1 \to V_2$ be injective and assume the coupling constants satisfy the bound
		\begin{equation}
			\beta_1 J_1(i,j) \leq \beta_2 J_2(\tau(i), \tau(j))
			\quad \text{for all $i,j \in V_1$}.
		\end{equation}
		Then
		\begin{equation}
			\sprod{\cos(\theta_x - \theta_y)}{G_1; \beta_1} 
			\leq \sprod{\cos(\theta_{\tau(x)} - \theta_{\tau(y)})}{G_2; \beta_2}
			\quad \text{for all $x,y \in V_1$}.
		\end{equation}
	\end{proposition}
	\begin{proof}
		Set $V_2' := \tau(V_1) \subset V_2$ and consider the graph $G_2' := (V_2, E'_2, J_2')$ where
		\begin{equation}
		J_2'(i,j) = \begin{cases}
		J_1(\tau^{-1}(i), \tau^{-1}(j)) & 	\quad\text{if}\quad i,j\in V_{2}'\\
                0 & 	\quad\text{otherwise},
		\end{cases}
		\end{equation}
	and $E'_2 := \set{\set{i,j} \subset V_2 \mid J'_2(i,j) > 0}$.
	With this definition $\beta_1 J_2' \leq \beta_2 J_2$ and hence
		\begin{equation}\label{eq:bound_G_G''}
			\sprod{\cos(\theta_{x'} - \theta_{y'})}{G_2'; \beta_1}
			\leq \sprod{\cos(\theta_{x'} - \theta_{y'})}{G_2; \beta_2}
			\quad \text{for all } x',y' \in V_{2},
		\end{equation}
		by an application of \zcref{thm:decrease_coupling}.
		We compute
		\begin{equation}
			\begin{aligned}
				Z_{G_2', \beta_1}
				&=\int_{[-\pi, \pi)^{V_2}} e^{{\beta_1} \sum_{\set{i,j} \in E_2'} J_2'(i,j) \cos \theta_i - \theta_j} \d \theta_{V_2} \\
				&=(2\pi)^{\abs{V_2 \setminus V_1}}\int_{[-\pi, \pi)^{\tau(V_1)}} e^{\frac{\beta_1} 2 \sum_{(i,j) \in V_1 \times V_1} J_2'(\tau(i), \tau(j)) \cos \theta_{\tau(u)} - \theta_{\tau(v)}} \d \theta_{\tau(V_1)} \\
				&=(2\pi)^{\abs{V_2 \setminus V_1}} \int_{[-\pi, \pi)^{V_1}} e^{\frac{\beta_1} 2 \sum_{(i,j) \in V_1 \times V_1} J_2'(i,j) \cos \theta_i - \theta_j} \d \theta_{V_1} 
				=Z_{G_1, \beta_1} 
			\end{aligned}
		\end{equation}
		and analogously
		\begin{equation}
			\begin{aligned}
				&=\int_{[-\pi, \pi)^{V_2}} \cos\left({\theta_{\tau(x)} - \theta_{\tau(y)}}\right)  e^{{\beta_1} \sum_{\set{i,j} \in E_2'} J_2'(i,j) \cos \theta_i - \theta_j} \d \theta_{V_2} \\
				&=(2\pi)^{\abs{V_2 \setminus V_1}} \int_{[-\pi, \pi)^{V_1}} \cos\left({\theta_{x} - \theta_{y}}\right) e^{\frac{\beta_1} 2 \sum_{(i,j) \in V_1 \times V_1} J_2'(i,j) \cos \theta_i - \theta_j} \d \theta_{V_1}.
			\end{aligned}
		\end{equation}
		Thus,
		\begin{equation}
			\sprod{\cos(\theta_x - \theta_y)}{G_1; \beta_1}
			= \sprod{\cos(\theta_{\tau(x)} - \theta_{\tau(y)})}{G_2'; \beta_1}
		\end{equation}
		for all $x, y \in V_1$. Together with \eqref{eq:bound_G_G''} this implies the claim.
	\end{proof}
	
	\subsection{The three-dimensional case}
	
	For the rest of this section, assume that $d \geq 3$. 
	
	\begin{proof}[Proof of \zcref{thm:main_interior} in $d \geq 3$]
		The proof is organized in  two steps.  
		First, we show that \eqref{eq:long_range_order_graph_interior_paralell_set} holds with $\Omega_{\delta}$ replaced by $ B(z; r_{z})$
for some $z \in \Omega$ and 
		\begin{equation}\label{eq:defrz}
			r_z := \sup \set{r > 0 \mid B(z; 10\, r) \subset \Omega}= \max \set{r > 0 \mid B(z; 10\, r) \subset \Omega}>0.
		\end{equation}
		Second, we will propagate this bound over finite distances by a finite cover of $\Omega_{\delta}$ with appropriately chosen balls.
		
		Let $\bar\eta > 0$ such that $[0, \bar\eta] \subset \supp \eta$ and set 
\begin{equation}\label{eq:hetadef}
h :=\frac{\bar\eta}{2} \frac{ \eps}{(2+\sqrt d)}.
\end{equation}
We choose $\eps$ small enough such that $\frac{\bar\eta}2 \eps \leq r_z$. In particular $h\leq  \frac{r_{z}}{(2+\sqrt d)}$.
		We consider two events $\mathcal A_1, \mathcal A_2 \subset \Xi$ defined by
		\begin{equation}\label{eq:A12altern}
			\begin{gathered}
				\mathcal A_1 := \set{\text{for all } v \in \Lambda_h \colon \abs{\V_n \cap \square_h(v)} \geq 1},\\
				\mathcal A_2 := \set{\text{for all } x \in \V_h \colon \deg_{n,\eps}(x) \leq nc_2},
			\end{gathered}
		\end{equation}
		where $c_2$ is the constant from \eqref{eq:degree_bound}. By \zcref{thm:each_square_has_k_points,thm:degrees_bounded} it holds that
		\begin{equation} 
			\P(\mathcal A_1 \cap \mathcal A_2) \geq 1 - c n e^{-c' n \eps^d}
		\end{equation}
		for positive constants $c, c'$ depending only on $\Omega$, $\eta$ and $\rho$. From now on, we consider $\omega$ in $\mathcal A_1 \cap \mathcal A_2$.\vspace{0,2cm}

		\emph{Step 1.} We define the rescaled set $\Lambda = \frac1h \Lambda_h$ and the graph $G = (\Lambda, E, J)$, where
		\begin{equation}
			E := \set{\set{i,j} \subset \Lambda \mid \abs{i-j}_1 = 1}
			\quad\text{and}\quad
			J(i,j) := \begin{cases}
				1 & \text{if } \set{i,j} \in E, \\
				0 & \text{else}.
			\end{cases}
		\end{equation}
		Note that $\Lambda$ defines a $d$-dimensional subset of $\Z^d$ and $J$ is a nearest neighbor interaction, so we can hope to apply \zcref{thm:nishimori_long_range_order}. 
		
		Let $z \in \Omega$ and recall the definitions of $r_{z}$  and $h$ \eqref{eq:defrz} \eqref{eq:hetadef}.  Fix $x,y \in \V_n(\omega) \cap B(z; r_z)$. We will apply \zcref{thm:comparison} with the choices
		\begin{equation}
			G_1 \equiv G \text{ and } G_2 = \G_{n,\eps}(\omega).
		\end{equation} 
		To define the injection $\tau \equiv \tau_{xy} \colon \Lambda \to \V_n(\omega)$ we proceed as follows. Let $\hat x, \hat y \in \Lambda$ such that
		\begin{equation}
			x \in \square_h(h \hat x) \quad \text{and} \quad y \in \square(h \hat y).
		\end{equation}
		If $\hat x = \hat y$ we redefine $\hat y$ to be any nearest neighbor of $\hat x$, i.e. $J(\hat x, \hat y) = 1$. 
		These points $\hat x, \hat y$ always exist, because $ \square_{2h}(x)\subset B(x;h \sqrt d) \subset B(z; r_z + h \sqrt d)\subset B(z; 2r_{z})\subset \Omega$.
		For $\hat v \in \Lambda$ we define
		\begin{equation}
			\tau(\hat v) := \begin{cases}
				x & \text{if } \hat v = \hat x,\\
				y & \text{if } \hat v = \hat y,\\
				\arg \min \set{\norm{z}_2 \mid z \in \V_n(\omega) \cap \square_h(h\hat{v})} & \text{else}.
			\end{cases}
		\end{equation}
		The choice for the points which are not equal to $\hat x$ or $\hat y$ is not important here, we could choose any point in $\V_n(\omega) \cap \square_h(h\hat{v})$. By construction this function is injective and 
		\[
		|\tau(\hat v)-h\hat v|\leq \frac{h\sqrt{d}}{2}\quad \forall \hat v\neq \hat y,\qquad |\tau(\hat y)-h\hat y|\leq h\left( \tfrac{\sqrt d}{2}+1\right).
		\]
Let $\hat u, \hat v \in \Lambda$ such that $J(\hat u, \hat v) = 1$, then 
		$\abs{\tau(\hat u) - \tau(\hat v)} \leq h (2+\sqrt{d})=\frac{\bar\eta}2\eps$ and therefore
		\begin{equation}
			\beta \mathcal J_{n,\eps}^\omega(\tau(\hat u), \tau(\hat v))
			= \frac{\beta \eta \left (\frac{\tau(\hat u)- \tau(\hat v)}{\eps}\right)}{\eps^d \sqrt{\deg_{n,\eps}(\tau(\hat u)) \deg_{n,\eps}(\tau(\hat v))}} 
			\geq \frac{\beta \eta\left(\frac{\bar\eta}2\right)}{c_2 n \eps^d}
			= \hat \beta_{n,\eps} J(\hat u, \hat v),
		\end{equation}
		where we used $\omega \in \mathcal{A}_{2}$ and we set
		\begin{equation}
			\hat \beta_{n,\eps} := \frac{\beta \eta\left(\frac{\bar\eta}2\right)}{c_2 n \eps^d}.
		\end{equation}
		Since $\mathcal J_{n,\eps}^\omega$ is non-negative we can therefore apply \zcref{thm:comparison} to obtain
		\begin{equation}
			\sprod{\cos \theta_x - \theta_y}{\G_{n,\eps}(\omega);\beta}
			=\sprod{\cos \theta_{\tau(\hat x)} - \theta_{\tau(\hat y)}}{\G_{n,\eps}(\omega);\beta}
			\geq \sprod{\cos \theta_{\hat x} - \theta_{\hat y}}{G;\hat\beta_{n,\eps}}.
		\end{equation}
		To apply \zcref{thm:nishimori_long_range_order} on the right-hand side, we need to prove that $B\left(\frac{\hat x + \hat y}2; 2 \abs{\hat x - \hat y}\right) \cap \Z^d \subset \Lambda$. Let $\hat w \in B\left(\frac{\hat x + \hat y}2; 2 \abs{\hat x - \hat y}\right) \cap \Z^d$. We show first
		$h\hat w \in\Omega$. Indeed using the definition of $r_z$ \eqref{eq:defrz}, we compute
		\begin{equation}
			\abs{h \hat w - z}
			\leq \frac52 \left(\abs{h \hat x - z} + \abs{h \hat y - z}\right)
			< 5 r_z + \frac52 \left(\abs{h \hat x - x} + \abs{h \hat y - y}\right) 
			\leq 5 (r_z + h (1+\sqrt d)).
		\end{equation}
Since we assumed that $h (1+\sqrt d) \leq \frac{\bar\eta}2\eps\leq  r_z$ we obtain that $\abs{h \hat w - z} < 10 r_z$ and therefore $h \hat w \in B(z; 10r_z) \subset \Omega$.
	 Therefore, by construction $h \hat w \in \Lambda_h$ and $\hat w \in \Lambda$.
		Hence, for $\hat\beta_{n,\eps} \geq \beta_*,$
		we obtain
		\begin{equation}
			\sprod{\cos \theta_x - \theta_y}{\G_{n,\eps}(\omega);\beta}
			\geq 1 - a_*\Gamma\left({\hat\beta_{n,\eps}}\right).
		\end{equation}
		where $a_*$ and $\beta_*$ are the constants from \zcref{thm:nishimori_long_range_order}.
		Adjusting the constants, this can be brought into the form of the statement \eqref{eq:long_range_order_graph_interior_paralell_set}.\vspace{0,2cm}

		\emph{Step 2.} By \zcref{thm:inner_parallel_set_connected} below, there exists a connected open set $U_\delta$ such that $\Omega_{\delta} \subset \overline{U_\delta} \subset \Omega$.
		Let $z_1, \ldots, z_{M-1} \in U_{\delta}$ such that
		\begin{equation*}
			\overline{U_\delta} \subset \bigcup_{k=1}^{M-1} B(z_k; r_{z_k}) =: W_{\delta} \subset \Omega.
		\end{equation*}
		Here $M \in \N$ is a constant that depends on $\Omega$ and $\delta$. Such a finite collection of points exists by compactness of $\overline{U_\delta}$. $W_\delta$ is connected because $U_\delta$ is connected.
		
		Define an undirected graph $T = (H, F)$, where $H = \set{z_1, \ldots, z_{M-1}}$ and
		\begin{equation}
			F := \set{\set{z, w} \subset H  \mid z \neq w \text{ and } B(z; r_z) \cap B(w; r_w) \neq \emptyset}.
		\end{equation}
		This graph is connected, which can be seen as follows (cf. \zcref{fig:3D_propagation}).
				There exists $j \in \set{2, \ldots, M-1}$ such that $B(z_1; r_{z_1}) \cap B(z_j; r_{z_j}) \neq \emptyset$, because otherwise we could write $W_\delta$ as the union of two disjoint, non-empty open sets $B(z_1; r_{z_1})$ and $\bigcup_{k=2}^{M-1} B(z_k; r_{z_k})$ which contradicts the connectivity of $W_\delta$. Without loss of generality we assume that $j = 2$. Then the vertices $z_1$ and $z_2$ are connected in $T$.
			Further, there exists $i \in \set{1,2}$ and $j \in \set{3, \ldots, M-1}$ such that $B(z_i; r_{z_i}) \cap B(z_j; r_{z_j}) \neq \emptyset$, because otherwise we would again obtain a contradiction to the connectivity of $W_\delta$. We assume without loss of generality that $j=3$. Thus all vertices in $\set{z_1, z_2, z_3}$ are (pairwise) connected in $T$. Iterating this procedure shows that $T$ is connected.
		
		For the rest of the proof we assume that $0\! <\! \eps \!\leq\! \eps_0(\Omega, \delta),$ with $\varepsilon_{0}$ small enough, such that $\bar\eta\eps_{0} \leq r_z$ for all $z \in H$ (cf. Step 1) and such that for any $\set{z,w} \in F$ there exists $a \in \Lambda_h$ with $\square_h(a) \subset B(z; r_z) \cap B(w; r_w)$.
		
		Let $x, y \in \V_n(\omega) \cap W_\delta$ and assume that $x \in B(z_x; r_{z_x})$ and $y \in B(z_y; r_{z_y})$ for $z_x, z_y \in H$. 		
		Let $z_x = w_1, w_2, \ldots, , w_{k-1}, w_k=z_y$ be a path in $T$ between $z_x$ and $z_y$ which is ordered such that $\set{w_i, w_{i+1}} \in F$ for all $i \in {1, \ldots, k-1}$. Note that $k \leq M-2$, because there are $M-1$ vertices in $H$. To any edge $\set{w_i, w_{i+1}}$ we associate some $\alpha_i \in \V_n(\omega) \cap \square_h(a_i)$, where $\square_h(a_i) \subset B(w_i; r_{w_i}) \cap B(w_{i+1}; r_{w_{i+1}})$. The point $\alpha_i$ exists, because we are considering realizations in the event $\mathcal A_{1}$ defined in \eqref{eq:A12altern}.
%		
%		
		%%%%%%%%%%%%%%%%%%%%%%%%%%%%%%%%%%%%%%%%%%%%%%%%%%%%%%%%%%%%%%%
		\begin{figure}
			\centering			
			\def\svgwidth{0.6\textwidth}
			\begingroup
			\makeatletter
			\newcommand*\fsize{\dimexpr\f@size pt\relax}
			\newcommand*\lineheight[1]{\fontsize{\fsize}{#1\fsize}\selectfont}
			\setlength{\unitlength}{\svgwidth}
			\global\let\svgwidth\undefined
			\makeatother
			\begin{picture}(1,0.40854425)
				\lineheight{1}
				\setlength\tabcolsep{0pt}
				\put(0,0){\includegraphics[width=\unitlength,page=1]{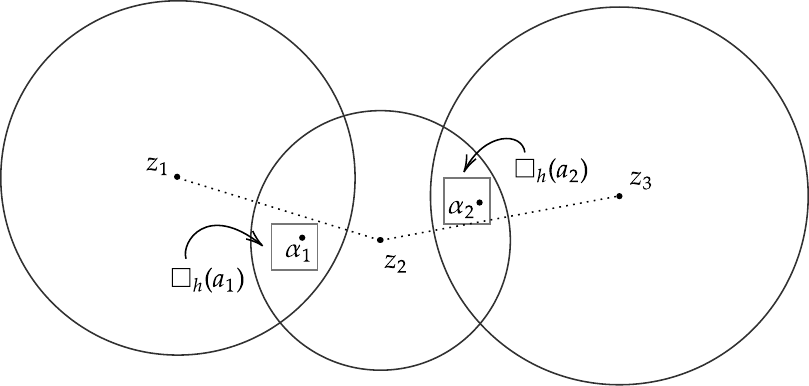}}
			\end{picture}
			\endgroup
			%\includesvg[width=0.6\textwidth]{img/3D_propagation}
			\caption{Balls $B(z_i; r_{z_i})$ for $i=1,2,3$. The dotted lines represent edges in $H$. The parameter $h$ is assumed small enough, such that $\square_h(a_1) \subset B(z_1; r_{z_1}) \cap B(z_2; r_{z_2})$ and $\square_h(a_2) \subset B(z_2; r_{z_2}) \cap B(z_3; r_{z_3})$. The points $\alpha_1, \alpha_2 \in \V_n$ are both contained in $B(z_2; r_{z_2})$.}
			\label{fig:3D_propagation} 
		\end{figure}
		Since two consecutive edges $e_i, ~e_{i+1}$  share a common endpoint $w_{i+1} \in H$, we have 
		\begin{equation}
			\alpha_i, \alpha_{i+1} \in B(w_{i+1}; r_{w_{i+1}})
		\end{equation}
		and therefore again by the first step
		\begin{equation}
			\sprod{\cos(\theta_{\alpha_i} - \theta_{\alpha_{i+1}})}{\G_{n,\eps}(\omega); \beta}
			\geq 1 - a_*\Gamma\left(\frac\beta{n\eps^d}\right)
		\end{equation}
		for all $i \in \set{1, \ldots, k-1}$. Further, since $
		x, \alpha_1 \in B(w_1; r_{w_1})$
		we have
		\begin{equation}
			\sprod{\cos(\theta_{x} - \theta_{\alpha_{1}})}{\G_{n,\eps}(\omega); \beta}
			\geq 1 - a_*\Gamma\left(\frac\beta{n\eps^d}\right)
		\end{equation}
		and since 
		$y, \alpha_{k-1} \in B(w_k; r_{w_k})$
		we have
		\begin{equation}
			\sprod{\cos(\theta_{y} - \theta_{\alpha_{k-1}})}{\G_{n,\eps}(\omega); \beta} 
			\geq 1 - a_*\Gamma\left(\frac\beta{n\eps^d}\right).
		\end{equation}
		Therefore, an application of \zcref{thm:scalar_product_triangle_expectation} below yields that
		\begin{equation}
			\sprod{\cos(\theta_{x} - \theta_{y})}{\G_{n,\eps}(\omega); \beta}
			\geq 1 - a_*(k+2)^2 \Gamma\left(\frac\beta{n\eps^d}\right)
			\geq 1 - a_* M^2 \Gamma\left(\frac\beta{n\eps^d}\right)
		\end{equation}
		as claimed with $a:= a_* M^2$.
	\end{proof}
	%%%%%%%%%%%%%%%%%%%%%%%%%%%%%%%%%%%%%%%%%%%%%%
	We conclude this subsection with two technical results used in the proof above.
	
	\begin{lemma}\label{thm:scalar_product_triangle_expectation}
		Let $\mu$ be a probability measure on $(\S^{n-1})^\Lambda$ with associated expectation $\sprod{\cdot}{}$. Suppose that $x_1, \ldots, x_L \in \Lambda$ satisfy
		\begin{equation}
			\sprod{S_{x_k} \cdot S_{x_{k+1}}}{} \geq 1 - \alpha \quad\text{for all } k \in \set{1, \ldots, L}
		\end{equation}
		for some $\alpha > 0$. Then
		\begin{equation}
			\sprod{S_{x_1} \cdot S_{x_{L}}}{} \geq 1 - L^2\alpha.
		\end{equation}
	\end{lemma}
	\begin{proof}
		We prove the statement by induction on $L$. It is clear for $L=1$. Thus assume it holds for $L-1$. We compute
		\begin{equation}
			\begin{aligned}
				\sprod{(S_{x_1} - S_{x_{L-1}})\cdot(S_{x_{L-1}}-S_{x_L})}{}
				&= \sprod{S_{x_1} \cdot S_{x_{L-1}}}{} + \sprod{S_{x_{L-1}} \cdot S_{x_{L}}}{} - 1 - \sprod{S_{x_1} \cdot S_{x_{L}}}{}\\
				&\geq 1 - (L-1)^2 \alpha - \alpha - \sprod{S_{x_1} \cdot S_{x_{L}}}{}.
			\end{aligned}
		\end{equation}
		On the other hand, by Cauchy--Schwarz
		\begin{equation}
			\begin{aligned}
				\abs{\sprod{(S_{x_1} - S_{x_{L-1}})\cdot(S_{x_{L-1}}-S_{x_L})}{}}
				&\leq \sqrt{\sprod{S_{x_1} - S_{x_{L-1}}}{} \sprod{S_{x_{L-1}} - S_{x_{L}}}{}} \\
				&= 2\sqrt{\sprod{1-S_{x_1} \cdot S_{x_{L-1}}}{} \sprod{1-S_{x_{L-1}} \cdot S_{x_{L}}}{}}\\
				&\leq 2 (L-1) \alpha.
			\end{aligned}
		\end{equation}
		Hence
		\begin{equation}
			\sprod{S_{x_1} \cdot S_{x_{L}}}{}
			\geq 1 - ((L-1)^2 + 1 + 2 (L-1)) \alpha
			= 1 - L^2 \alpha
		\end{equation}
		concluding the proof.
	\end{proof}
	
	\begin{lemma}\label{thm:inner_parallel_set_connected}
		Let $\Omega \subset \R^d$ be a connected, open domain with Lipschitz boundary. For all $\delta > 0$ there exists an open connected subset $U_\delta \subset \Omega$ such that
		\begin{equation}
			\Omega_\delta \subset \overline{U_\delta} \subset \Omega,
		\end{equation}
		where $\Omega_\delta := \set{x \in \Omega \mid \dist(x, \partial \Omega) > \delta}$.
	\end{lemma}
	\begin{proof}
		It suffices to prove the claim for all $\delta \leq \delta_0$ for some fixed $\delta_0 > 0$, because if $\delta > \delta_0$, then
		\begin{equation}
			\Omega_\delta \subset \Omega_{\delta_0},
		\end{equation}
		so we can choose $U_\delta := U_{\delta_0}$.
		
		If $\Omega$ has a $C^{2,\alpha}$ boundary, then it follows from \cite[Lemma 2.6]{Fonseca:SecondorderGammalimit2026} that $\Omega_\delta$ is connected for all $0 < \delta \leq \delta_0$, where $\delta_0(\Omega) > 0$ is a constant which depends only on the domain, so that $U_\delta := \Omega_\delta$. 
		
		If $\Omega$ is a general Lipschitz domain, there exists a bounded $C^\infty$ open domain $\hat \Omega$ and $\tau \colon \Omega \to \hat\Omega$ bi-Lipschitz \cite[Remark 5.3]{Ball:PartialRegularitySmooth2017}. Moreover, we can extend $\tau \colon \overline \Omega \to \overline{\hat\Omega}$ such that this remains true and $\tau(\partial \Omega) = \partial \hat\Omega$. Therefore, there exists $\eps_0 > 0$ such that $\hat \Omega_\eps$ is connected for all $\eps \leq \eps_0$.
		
		Let $\delta > 0$. We claim that $\tau(\Omega_\delta) \subset \hat \Omega_{c_\tau \delta}$ for a constant $c_\tau$ depending only on $\tau$. Indeed, for any $x \in \Omega_\delta$ and $y \in \partial \hat \Omega$ we have $\tau^{-1}(y)\in \partial  \Omega$ and hence
		\begin{equation}
			\abs{\tau(x) - y} 
			= \abs{\tau(x)- \tau(\tau^{-1}(y))}
			\geq c_{\tau} \abs{x - \tau^{-1}(y)}
			> c_{\tau} \delta
		\end{equation}
		which implies $\tau(x) \in \hat\Omega_{c_\tau \delta}$.
		If $c_\tau \delta \leq \eps_0$, we set $U_\delta := \tau^{-1}(\hat\Omega_{c_\tau \delta})$ which is connected by definition of $\eps_0$ and since $\tau$ is continuous. With this definition $\Omega_{\delta }\subset U_{\delta }$. Moreover, we have that $\overline{U_\delta} \subset \Omega$ because
		\begin{equation}
			x \in \overline{U_\delta}
			\Rightarrow \tau(x) \in \overline{\hat \Omega_{c_{\tau }\delta}} \subset \hat \Omega
			\Rightarrow x = \tau^{-1}(\tau(x)) \in \Omega,
		\end{equation}
		where the first implication follows by approximating $x$ by a sequence $(x_n)_{n \in \N} \subset U_\delta$ and using (sequential) continuity of $\tau$.
	\end{proof}
	
	\subsection{Modifications in two dimensions}
	
	To prove \zcref{thm:main_interior} in two dimensions, we need to modify the first step in the proof of the three dimensional case. The reason is that because of the Mermin--Wagner theorem there is no long range order of the \emph{$XY$}-model on the lattice and thus \zcref{thm:nishimori_long_range_order} does not apply. Instead, we use that the probability to find many (order $\mathcal O(h^{-1})$) random points in each $h$-square is large. These points are used to build a ``third dimension'' artificially, thus enabling the application of \zcref{thm:nishimori_long_range_order}.
	
	As before, let $\bar\eta > 0$ such that $[0, \bar\eta] \subset \supp \eta$ and set $h := \frac{\bar\eta\eps}{2 (1+\sqrt{2})}$. Further, we let
	\begin{equation}
		k :=\lfloor 12 \diam \Omega / h\rfloor.
	\end{equation}
	If $k$ is odd we redefine $k \mapsto k+1$ without relabeling. We consider two events $\mathcal A_1, \mathcal A_2 \subset \Xi$ defined by
	\begin{equation}
		\begin{gathered}
			\mathcal A_1 := \set{\text{for all } v \in \Lambda_h \colon \abs{\V_n \cap \square_h(v)} \geq k},\\
			\mathcal A_2 := \set{\text{for all } x \in \V_h \colon \deg_{n,\eps}(x) \leq nc_2},
		\end{gathered}
	\end{equation}
	where $c_2$ is the constant from \eqref{eq:degree_bound}.
	By \zcref{thm:each_square_has_k_points,thm:degrees_bounded}, we have that
	\begin{equation}
		\P(\mathcal A_1 \cap \mathcal A_2) \geq 1 - c n e^{-c' n \eps^2},
	\end{equation}
	provided that $n \eps^3$ is large enough. In the following, we thus consider $\omega \in \mathcal A_1 \cap \mathcal A_2$.
	
	This time, set $\Lambda := \frac1h \Lambda_h \times \set{1, \ldots, k} \subset \Z^3$ and set for $i,j \in \Lambda$
	\begin{equation}
		E := \set{\set{i,j} \subset \Lambda \mid \abs{i-j}_1 = 1}
		\quad\text{and}\quad
		J(i,j) := \begin{cases}
			1 & \text{if } \set{i,j} \in E, \\
			0 & \text{else},
		\end{cases}
	\end{equation}
	where we view $\Lambda$ as embedded in $\Z^3$ with the $\abs{\cdot}_1$ being the $\ell^1$-norm on the three dimensional lattice. As before, we wish to apply \zcref{thm:comparison} with 
	\begin{equation}
		G_1 = (\Lambda, E, J)
		\text{ and }
		G_2 \equiv \G_{n,\eps}(\omega).
	\end{equation}
	We need to construct an injection $\tau \colon \Lambda \to \V_n(\omega)$.
	Fix $z \in \Omega$ and define 
	\begin{equation}\label{eq:def_rz_dim2}
		r_z := \sup \set{r > 0 \mid B(z;10r_z) \subset \Omega},
	\end{equation}
	and choose $\eps$ small enough such that $\frac{\bar\eta}2\eps \leq r_z$. We show that \eqref{eq:long_range_order_graph_interior_paralell_set} holds with $\Omega_\delta$ replaced by $B(z, r_z)$. Fix $x,y \in \V_n(\omega) \cap B(z; r_z)$ and let $x', y' \in \Lambda_h$ be the unique points such that
	\begin{equation}
		x \in \square_h(x') \quad \text{and} \quad y \in \square_h(y').
	\end{equation}
	These points always exist by the same arguments as in the case $d\geq 3$.
	We define now 
	\begin{equation}
		\hat x:= \left(h^{-1}x', k/2+1\right) \in \Lambda
		\quad\text{and}\quad
		\quad \hat y:=\left(h^{-1}y', k/2\right)\in \Lambda.
	\end{equation}
	With this definition $\hat x\neq \hat y$ holds always.
	For any $v' \in \Lambda_h$ choose $k$ distinct points $v_1, \ldots, v_k \in \V_n(\omega) \cap \square_h(v')$ by some arbitrary rule (for example, by taking the $k$ points with the smallest $\ell^1$-norm). These points exist because of the assumption $\mathcal A_{1}$ on our random environment.
	Given $\hat v = (h^{-1}v', q) \in \Lambda$, $v' \in \Lambda_h$ and $q\in \set{1, \ldots, k}$, define $\tau \equiv \tau_{xy} \colon \Lambda \to \V_n(\omega)$ by
	\begin{equation}
		\tau(\hat v)=	\tau(h^{-1}v', q) := \begin{cases}
			x & \text{if } \hat v = \hat x, \\
			y & \text{if } \hat v = \hat y, \\
			v_{q} & \text{else}.
		\end{cases}
	\end{equation}
	By construction $\tau$ is injective. 
	
	Let now $\hat u = (h^{-1} u', p), \hat v = (h^{-1} v', q) \in \Lambda$ such that $J(\hat u, \hat v)=1$. Then we have either $\abs{u' - v'}_1 = h$ and $p = q$, or $u' = v'$ and $\abs{p - q}=1$. In the first case $\tau(\hat u) = u_{p}$ and $\tau(\hat v) = v_{q}$ belong to two different $h$-boxes which share a common one dimensional face, so that 
	\begin{equation}
		\abs{\tau(\hat u) - \tau(\hat v)}
		\leq \abs{\tau(\hat u) - u'} + \abs{v' - \tau(\hat v)} + \abs{u' - v'}
		\leq h (1+ \sqrt 2) = \frac{\bar\eta}2\eps,
	\end{equation}
	while in the second case $\tau(\hat u), \tau(\hat v) \in \square_h(u')$ and thus
	\begin{equation}
		\abs{\tau(\hat u) - \tau(\hat v)}
		\leq \abs{\tau(\hat u) - u'} + \abs{u' - \tau(\hat v)}
		\leq h \sqrt 2 \leq \frac{\bar\eta}2\eps.
	\end{equation}
	Hence, using \zcref{ass:eta} and that $\omega \in \mathcal{A}_{2},$
	\begin{equation}
		\beta \mathcal J_{n,\eps}^\omega(\tau(\hat u), \tau(\hat y)) 			
		= \frac{\beta\, \eta\left (\frac{\tau(\hat u)- \tau(\hat v)}{ \eps}\right )}{\eps^2 \sqrt{\deg_{n,\eps}(\tau(\hat u))\deg_{n,\eps}(\tau(\hat v))}} 
		\geq \frac{\beta \eta\left(\frac{\bar\eta}2\right)}{c_2 n \eps^2}
		= \hat \beta_{n,\eps} J(\hat u, \hat v).
	\end{equation}
	As in the proof for $d \geq 3$, we apply \zcref{thm:comparison} to obtain
	\begin{equation}
		\sprod{\cos \theta_x - \theta_y}{\G_{n, \eps}(\omega);\beta}
		=\sprod{\cos \theta_{\tau(\hat x)} - \theta_{\tau(\hat y)}}{\G_{n,\eps}(\omega);\beta}
		\geq \sprod{\cos \theta_{\hat x} - \theta_{\hat y}}{G;\hat\beta_{n,\eps}}.
	\end{equation}
	We can apply \zcref{thm:nishimori_long_range_order} if $B\left(\frac{\hat x + \hat y}2; 2 \abs{\hat x - \hat y}\right)\cap \Z^{3} \subset \Lambda,$
	where 
	\begin{equation}
		\frac{\hat x + \hat y}2= \left(\frac{x'+y'}{2h},  \frac{k+1}{2} \right),\quad
		\hat x - \hat y= \left(h^{-1} ( x'-y'), 1\right).
	\end{equation}
	Let $\hat w = (h^{-1}w',q)\in B\left(\frac{\hat x + \hat y}2; 2 \abs{\hat x - \hat y}\right)\cap \Z^{3}$.  On the one hand,
	using $r_{z}\leq \diam \Omega$, we argue
	\begin{equation}
		\begin{aligned}
			\abs{q - \frac{k+1}2}
			&\leq \abs{\hat w - \frac{\hat x + \hat y}2} 
			\leq 2 \abs{\hat x - \hat y} 
			= \frac 2h \abs{x' - y'}+2
			\leq 2(\sqrt 2 + 1) + \frac{4 r_z}{h}
			\leq 2(\sqrt 2 + 1) + \frac{k}3
		\end{aligned}
	\end{equation}
	by \eqref{eq:def_rz_dim2}. Thus $\bar w \in \set{1, \ldots, k}$ for $\eps$ small enough (and consequently $k$ large).
	
	On the other hand we have
	\begin{equation}
		\begin{aligned}
			\abs{w' - z}
			&\leq \abs{w' - \frac{x'+y'}2} + \frac12(\abs{x' - z} + \abs{y' - z})
			\leq \frac52 (\abs{x' - z} + \abs{y' - z})\\
			&\leq 5(r_z + \tfrac{h\sqrt d}2)
			< 10 r_z,
		\end{aligned}
	\end{equation}
	where we used that $h\sqrt d \leq \frac{\bar\eta}2\eps \leq r_z$ in the last line. Thus $w' \in B(z; 10r_z) \subset \Omega$ and hence $w' \in \Lambda_h$. Therefore, $\hat w \in \Lambda$.
	We can thus apply \zcref{thm:nishimori_long_range_order} and obtain for $\hat\beta_{n,\eps} \geq \beta_*$, where $\beta_*$ is the constant from \zcref{thm:nishimori_long_range_order},
	\begin{equation}
		\sprod{\cos \theta_x - \theta_y}{\G_{n, \eps}(\omega);\beta}
		\geq \sprod{\cos \theta_{\hat x} - \theta_{\hat y}}{G;\hat\beta_{n,\eps}}
		\geq 1 - a_*\Gamma\left(\frac\beta{n\eps^d}\right).
	\end{equation}
	For the second step we proceed exactly as in the case when $d \geq 3$. \qed
	
	\section*{Acknowledgments}
	
	It is our pleasure to thank Tom Spencer for sharing his many insights on the model. We also thank Tyler Helmuth for useful discussions
	on the strategy of some proofs.
	This reasearch was supported by the German Research Foundation (Deutsche Forschungsgemeinschaft, DFG)
	under Germany’s Excellence Strategy-GZ 2047/1 - 390685813,  by the DFG -- CRC 1720 -- 539309657 and by the DFG -- CRC 1060 -- 211504053.

	%%%%%%%%%%%%%%%%%%%%%%%%%%%%%%%%%%%%%%%%%%%%%%%%%%%%%%%%%%%%%%%%%%%%%%%%%%%%%%%%%%%%%%%%%%%%%%%%%%%%
	\appendix
	
	\section{Proofs on cycle spaces}
	\label{appendix}
	
	We provide here the proof of \zcref{thm:lifting_basis}. We prove first that a lifting exists.
	
	\begin{lemma}\label{thm:lifting_exists}
		Let $G = (V,E)$ be an undirected graph with orientation $\vec G$.  Let $f \in \mathbf C(G)$. Then there exists a lifting $\vec f \in \mathbf C(\vec G)$ of $f$.
	\end{lemma}
	\begin{proof}
		We assume that $f$ is connected (as a collection of edges in $E$), otherwise we lift all connected components individually. Then $f$ is an Eulerian graph, i.e. every vertex has even degree. Thus, there exists a sequence of  distinct edges $(e_0, \ldots e_n)$, $e_k = \set{v_{k}, v_{k+1}} \in E$ with $v_0 = v_{n+1}$, such that $	f = \sum_{k = 1}^n \I_{e_k}$.
		Set
		\begin{equation}
			\vec f(\vec e) := \sum_{k=0}^n \delta_{\vec e, (v_{k}, v_{k+1})} - \delta_{\vec e, (v_{k+1}, v_{k})}, \quad \vec e \in \vec E.
		\end{equation}
		Then for all $v \in V$ 
		\begin{equation}
			\begin{aligned}
				&
				\sum_{\vec e \in \delta^+(v)} \hspace{-0.2cm} \vec f(\vec e) -\hspace{-0.2cm} \sum_{\vec e \in \delta^-(v)} \hspace{-0.2cm}\vec f(\vec e) 
				= \sum_{k=0}^n \left(\sum_{\vec e \in \delta^+(v)} \left(\delta_{\vec e, (v_{k}, v_{k+1})} - \delta_{\vec e, (v_{k+1}, v_{k})}\right) -\hspace{-0.2cm}  \sum_{\vec e \in \delta^-(v)} \left(\delta_{\vec e, (v_{k}, v_{k+1})} - \delta_{\vec e, (v_{k+1}, v_{k})}\right)\right) \\
				&\quad =\sum_{k=0}^n \left(\sum_{\vec e \in \delta^+(v)} \delta_{\vec e, (v_{k}, v_{k+1})} + \sum_{\vec e \in \delta^-(v)} \delta_{\vec e, (v_{k+1}, v_{k})} - \sum_{\vec e \in \delta^+(v)} \delta_{\vec e, (v_{k+1}, v_{k})} - \sum_{\vec e \in \delta^-(v)} \delta_{\vec e, (v_{k}, v_{k+1})}\right)
			\end{aligned}
		\end{equation}
		Further we have
		\begin{equation}
			\sum_{\vec e \in \delta^+(v)} \delta_{\vec e, (v_{k}, v_{k+1})} + \sum_{\vec e \in \delta^-(v)} \delta_{\vec e, (v_{k+1}, v_{k})} = \delta_{v, v_k},\quad 	\sum_{\vec e \in \delta^+(v)} \delta_{\vec e, (v_{k+1}, v_{k})} + \sum_{\vec e \in \delta^-(v)} \delta_{\vec e, (v_{k}, v_{k+1})} = \delta_{v, v_{k+1}}.
		\end{equation}
		Thus,
		\begin{equation}
			\sum_{\vec e \in \delta^+(v)} \vec f(\vec e) - \sum_{\vec e \in \delta^-(v)} \vec f(\vec e) = \sum_{k=0}^n \left(\delta_{v, v_k} - \delta_{v, v_{k+1}}\right) = 0,
		\end{equation}
		because $v_0 = v_{n+1}$.
		Hence $\vec f \in \mathbf C(\vec G)$. 
		
		Since the $e_k$ are distinct, we have that $\vec f \in \set{-1, 0, 1}$ and $\vec f(\vec e) \neq 0$ if and only if $e \in \set{e_0, \ldots, e_n}$. Therefore, $\pi_{\vec f} = f$ and $\vec f$ is a lifting of $f$.
	\end{proof}
	\begin{proof}[Proof of \zcref{thm:lifting_basis}]
		By \zcref{thm:lifting_exists} the $\vec f_1, \ldots, \vec f_n$ always exists. We first show that they are linearly independent. Assume the contrary was true, i.e. there exist $\lambda_1, \ldots, \lambda_n \in \Q$ not all zero such that 
		\begin{equation}
			0 = \sum_{k=1}^n \lambda_k \vec f_k.
		\end{equation}
		Multiplying both sides with a constant, we may assume that $\lambda_k \in \Z$ for all $k=1, \ldots, n$ and at least one is odd. Thus
		\begin{equation}
			0 = \sum_{k=1}^n \left(\lambda_k \bmod 2\right) \pi_{\vec f_k} = \sum_{k=1}^n \left(\lambda_k \bmod 2\right) f_k
		\end{equation} 
		and at least one $\lambda_k \mod 2 \neq 0$. But this contradicts the $f_k$ being a basis.
	
			Let now $F$ be a spanning forest of $G$ and let $N \subset E$ be all the edges which do not belong to $F$.
			We show the identity
		\begin{equation}
			\dim \mathbf C(G) =|N|=\dim \mathbf C(\vec G),
		\end{equation}
		which implies the proposition.
		For $e = \set{i,j} \in E$ let $P_e$ be the unique path in $F$ connecting $i$ and $j$. Define for $e'\in E$
		\begin{equation}
			g_e(e') := \begin{cases}
				1 & \text{if } e' \in P_e \text{ or } e = e', \\
				0 & \text{else.}
			\end{cases}
		\end{equation}
		Then $g_e \in \mathbf C(G)$ and for all $e, e' \in N$ we have $	g_{e}(e') = \delta_{e, e'}$.
		Thus, the $\left(g_e\right)_{e \in N}$ are linearly independent. We claim they span $\mathbf C(G)$. Indeed, let $h \in \mathbf C(G)$. Define
		\begin{equation}
			h' := \sum_{e \in N} h(e) g_e \in \mathbf C(G).
		\end{equation}
		Let $H := h -h'$. If $e' \in N$, then
		\begin{equation}
			H(e') = h(e') - \sum_{e \in N} h(e) \delta_{e,e'} = h(e') -h(e') = 0.
		\end{equation}
		Thus $H$ is a cycle using only edges of a forest, so we must have $H \equiv 0$ and in particular $h' = h$. Therefore, the $\left(g_e\right)_{e \in N}$ span the whole space and $\dim \mathbf C(G) = \abs{N}$.
		
		Let $\vec g_e \in \mathbf C(\vec G)$ be a lifting of $g_e$. By multiplying with a constant, we assume $\vec g_e(\vec e) = 1$. By the first part, we already know that the $\left(\vec g_e\right)_{e \in N}$ are linearly independent. We claim that they span $\mathbf C(\vec G)$. Let $\vec h \in \mathbf C(\vec G)$ and set
		\begin{equation}
			\vec h' := \sum_{e \in N} h(\vec e) \vec g_e \in  \mathbf C(\vec G).
		\end{equation}
		Let $\vec H := \vec h - \vec h'$. If $e' \in N$, then
		\begin{equation}
			\vec H(\vec e') = \vec h(\vec e') - \sum_{e \in N} \vec h(\vec e) \delta_{e,e'} = \vec h(\vec e') - \vec h(\vec e') = 0.
		\end{equation}
		Since $\vec H$ can be non-zero only on the edges of a forest, this implies that $\vec H \equiv 0$. Hence, $\dim \mathbf C(\vec G) = \abs{N}$.
	\end{proof}

\end{document}